\documentclass{easychair}

\usepackage[T1]{fontenc}
\usepackage{amsmath, amssymb}
\usepackage{pifont}
\usepackage{subcaption}
\usepackage{url}
\usepackage{hyperref}
\usepackage{doi}
\usepackage{multirow}
\usepackage{array}
\usepackage{booktabs}
\usepackage[numbers]{natbib}
\usepackage[dvipsnames]{xcolor}
\usepackage{forest}
\usepackage{enumitem}
\usepackage{orcidlink}
\usepackage{ifthen}

\newboolean{appendix}
\setboolean{appendix}{true}

\newcommand{\appendixreference}[2]{%
  \ifthenelse{\boolean{appendix}}%
    {Appendix~\ref{#1}}%
    {\citep[Appendix #2]{ZamponiLBC26}}%
}

\usepackage{todonotes}
\newcolumntype{R}[1]{>{\raggedleft\let\newline\\\arraybackslash}p{#1}}
\newcommand*{\stlsat}{\textsc{STLSat}}

\newcommand*{\set}[1]{\{ #1 \}}

\newcommand*{\marked}[1]{\overline{#1}}
\newcommand*{\stlunop}[3][]{\mathop{\mathsf{#2}^{#1}_{#3}}}
\newcommand*{\stlbinop}[5][]{{#4} \mathbin{{#2}^{#1}_{#3}} {#5}}
\newcommand*{\evenop}{\mathsf{F}}
\newcommand*{\globop}{\mathsf{G}}
\newcommand*{\untilop}{\mathsf{U}}
\newcommand*{\releaseop}{\mathsf{R}}
\newcommand*{\suntilop}{\mathsf{sU}}
\newcommand*{\sreleaseop}{\mathsf{sR}}
\newcommand*{\until}[3]{\stlbinop{\untilop}{#1}{#2}{#3}}
\newcommand*{\release}[3]{\stlbinop{\releaseop}{#1}{#2}{#3}}
\newcommand*{\even}[2][]{\stlunop[{#1}]{\evenop}{#2}}
\newcommand*{\glob}[2][]{\stlunop[{#1}]{\globop}{#2}}
\newcommand*{\meven}[2][]{\stlunop[{#1}]{\marked{\evenop}}{#2}}
\newcommand*{\mglob}[2][]{\stlunop[{#1}]{\marked{\globop}}{#2}}
\newcommand*{\suntil}[4][]{\stlbinop[#1]{\suntilop}{#2}{#3}{#4}}
\newcommand*{\srelease}[4][]{\stlbinop[#1]{\sreleaseop}{#2}{#3}{#4}}
\newcommand*{\anybinop}[4][]{\stlbinop[{#1}]{\mathsf{B}}{#2}{#3}{#4}}
\newcommand*{\msuntil}[4][]{\stlbinop[#1]{\marked{\suntilop}}{#2}{#3}{#4}}
\newcommand*{\msrelease}[4][]{\stlbinop[#1]{\marked{\sreleaseop}}{#2}{#3}{#4}}
\newcommand*{\anyunop}[2][]{\stlunop[#1]{A}{#2}}

\renewcommand{\implies}{\rightarrow}

\renewcommand*{\exp}{\operatorname{tex}}

\newcommand{\cmark}{\textcolor{ForestGreen}{\ding{51}}}
\newcommand{\xmark}{\textcolor{red}{\ding{55}}}

\newcommand*{\pactive}{P}

\newcommand{\tinv}[0]{\mathcal{T}}

\newcommand*{\oldjump}{$\mathsf{JUMP}_\mathsf{Old}$}
\newcommand*{\stltree}{\textsc{STLTree}}

\newcommand{\parof}[2]{\langle #1\rangle_{#2}}

\newcommand{\id}[1]{\text{id}(#1)}

\newtheorem{theorem}{Theorem}
\newtheorem{lemma}[theorem]{Lemma}
\newtheorem{corollary}[theorem]{Corollary}

\begin{document}

\title{\stlsat{}---An Improved Tableau for Satisfiability Checking of Signal Temporal Logic Formulas}
\titlerunning{\stlsat{}}

\author{%
  Marco Zamponi\inst{1} \orcidlink{0009-0002-1855-4469} \and
  Florian Lammel\inst{2} \and
  Ezio Bartocci\inst{2} \orcidlink{0000-0002-8004-6601} \and
  Michele Chiari\inst{2,3} \orcidlink{0000-0001-7742-9233}}

\institute{%
  IMT Lucca, Lucca, Italy, \email{name.surname@imtlucca.it} \and
  TU Wien, Vienna, Austria, \email{name.surname@tuwien.ac.at} \and
  AIT Austrian Institute of Technology, Vienna, Austria,
  \email{name.surname@ait.ac.at}
}

\authorrunning{Zamponi, Lammel, Bartocci and Chiari}

\maketitle

\begin{abstract}
    Signal Temporal Logic (STL) is a formalism used to describe temporal properties of real-valued signals in cyber-physical systems.
    In mission- and safety-critical domains, specifications often consist of large collections of STL formulas, making consistency checking and requirement analysis a major engineering bottleneck.
    Tableau-based satisfiability procedures are a natural way to address this problem.
    Two of the authors of this paper contributed to the only existing tree-shaped tableau for bounded discrete-time STL, but we have recently found out that the procedure can return incorrect verdicts for some STL formulas.
    In this paper, we pinpoint the flaw in that procedure and present a new tree-shaped tableau that we prove to be \textbf{sound} and \textbf{complete} for bounded discrete-time STL.  
    
    On top of this theoretical foundation, we introduce \stlsat{}, an open-source Rust tool that decides the satisfiability of STL formulas, synthesizes concrete witness signals, and extracts unsatisfiable cores with its tableau engine, allowing users to identify inconsistent subsets of requirements for more effective specification debugging.  \stlsat{} also implements enhanced first-order logic and satisfiability modulo theories encodings for STL, which allow it to act as a portfolio solver.

    We evaluate \stlsat{} on an extended benchmark suite (including STL and Mission-time Linear Temporal Logic formulas) that we release publicly.  Across the whole benchmark, the portfolio solver matches or outperforms state-of-the-art tools while preserving correctness.
\end{abstract}

\section{Introduction}
\label{sec:intro}

Signal Temporal Logic (STL)~\citep{MalerN04} is a well-established specification language for describing the temporal properties of real-valued signals in mission- and safety-critical cyber-physical systems (CPS).  STL underpins a broad spectrum of verification tasks, ranging from runtime monitoring~\citep{JaksicBGKNN15} and mutation testing \citep{BartocciMNY23} to falsification \citep{SankaranarayananF12}, parameter synthesis \citep{BartocciBNS15}, fault localization~\citep{BartocciFMN18} and debugging \citep{BartocciMMMN21}.
During the design of a CPS, engineers typically collect a large number of STL requirements.
In this context, the following two situations can be problematic:

\noindent \textbf{Inconsistency.} 
    A set of STL formulas $\Phi=\{\varphi_1,\dots,\varphi_n\}$ is \emph{consistent}
    iff there exists at least one signal $s$ such that 
    $s \models \bigwedge_{i=1}^{n}\varphi_i$.  Equivalently, the conjunction
    $\bigwedge_{i=1}^{n}\varphi_i$ must be \textbf{satisfiable}; if it is not,
    the requirements are mutually contradictory.

\noindent \textbf{Redundancy.}  
    A formula $\psi\in\Phi$ is \emph{redundant} w.r.t.\ the remaining
    requirements iff those requirements already entail $\psi$, i.e., 
    $\bigwedge_{\varphi\in\Phi\setminus\{\psi\}} \models_{\mathsf E} \psi$.
    By definition, this holds exactly when the formula
    $\bigl(\bigwedge_{\varphi\in\Phi\setminus\{\psi\}}\bigr)\land\neg\psi$ is
    \textbf{unsatisfiable}.

Moreover, manual analysis of large CPS requirement sets is error-prone and demanding in terms of time, expertise, and cognitive effort, especially when expressed in a formal language such as STL.
Thus, automated requirement engineering tools are highly desirable to support engineers in these tasks.
Luckily, both notions of consistency and redundancy can be reduced to a single STL-satisfiability query: for any pair of formulas $\varphi,\psi$, the implication $\varphi\implies\psi$ holds iff $\varphi\land\neg\psi$ is unsatisfiable, and equivalence follows from mutual implication.
Hence, in this paper we tackle the problem of checking STL satisfiability.

These reductions are especially valuable in \emph{specification mining}~\cite{BartocciMNN22}, where automatically extracted temporal-logic specifications may contain noise, redundancies, or inconsistencies.
Satisfiability checking can detect inconsistent subsets, eliminate redundancies, and cluster semantically equivalent specifications, yielding a more concise and traceable requirement model.



\smallskip
\noindent \textbf{Example.}
Consider an autonomous crop-irrigation controller whose water-flow signal $f$
is measured along the irrigation line.
The following three STL requirements have been elicited:
\begin{align*}
\varphi_1 &:= \glob{[0,100]}\,\even{[10,40]} (f \ge 50) \\
\varphi_2 &:= \glob{[0,120]}\bigl( f \ge 50 \;\Rightarrow\;
               \even{[0,10]}\,\glob{[0,30]} (f < 5) \bigr) \\
\varphi_3 &:= \glob{[0,120]} (f < 100)
\end{align*}

\begin{itemize}[nosep]
  \item $\varphi_1$ requires that, at every time in \([0,100]\), the flow reaches at least \(50\) l/min within the following 10–40 time units.
  \item $\varphi_2$ states that whenever the flow reaches this irrigation value, it must be reduced below $5$ l/min within the next $10$ time units and stay low for another $30$ time units,
 in order to avoid over-watering.
  \item $\varphi_3$ is a safety bound that forbids flows equal to or above
        $100$ l/min at any time in the interval 0-120 units, preventing flooding of the field.
\end{itemize}

The three formulas together form the specification  
$\Phi_w=\{\varphi_1,\varphi_2,\varphi_3\}$.
Before $\Phi_w$ can be used for model checking or falsification, we must verify that it is \textbf{consistent}, i.e.\ that the
conjunction $\bigwedge_{i=1}^{3}\varphi_i$ is \emph{satisfiable}.  
Intuitively, $\varphi_2$ forces a drop from $f\ge 50$ to $f<5$,
which might conflict with $\varphi_1$;
however, turning this intuition into a rigorous proof requires exploring all possible signal behaviors over $[0,140]$.
Manual reasoning is error-prone, whereas automated STL satisfiability tools decide consistency in seconds; if the set is satisfiable they may synthesize a concrete witness signal, and if it is unsatisfiable they may extract an unsatisfiable core that pinpoints the conflicting requirements, enabling their immediate refinement.

Although STL has become the \emph{de facto} standard specification language for CPS requirements, it is still lacking mature, high-performance, and dedicated satisfiability solvers.
Existing approaches either target different logics---e.g., MLTLSAT \citep{mltlsat}---or exist only as research prototypes such as the tableau implementation described in \citet{MelaniBC25}. 
Moreover, these tools do not integrate smoothly into verification pipelines, which limits their usefulness for practical requirement-engineering workflows that require automated debugging and traceability.


The most prominent tableau-based solver for bounded discrete-time STL is \stltree{}~\citep{MelaniBC25}, also co-developed by two of the authors of this paper.
Its construction accelerates reasoning by ``jumping'' over intervals, but a careful analysis reveals a soundness flaw: in formulas with nested temporal operators, the old jump rules may skip instants where conflicts would be introduced, causing the tableau to miss rejecting branches.
Consider the following unsatisfiable formula:
\[
\varphi_S :=
\glob{[0,5]} \bigl(p \lor \glob{[2,2]}p\bigr) \land \glob{[3,5]}q,
\]
where $p$ and $q$ are inconsistent atomic predicates, for example $p = x \geq 0$ and $q = x < 0$.
At time 3, the second conjunct requires $q$; hence, the first conjunct cannot be satisfied by choosing $p$ and instead generates the obligation $\glob{[5,5]}p$.
The second conjunct also requires $q$ at time 5, making $\varphi_S$ unsatisfiable.
However, \stltree{} jumps from time 3 to time 5 and incorrectly shifts the obligation from $\glob{[5,5]}p$ to $\glob{[7,7]}p$.
The resulting tableau avoids the conflict between $p$ and $q$ at time 5 and accepts $\varphi_S$ with a signal that is not actually a model of $\varphi_S$.

This motivates our first contribution: a corrected tableau algorithm with formally proved soundness and completeness. Our approach inspects proposition-appearance intervals to guarantee that no conflicts can be introduced in skipped time steps.

Building on this foundation we introduce \stlsat{}, an open-source Rust tool implementing three complementary decision procedures that can be executed in parallel to speed up the solving process: our correct tableau, an optimized first-order-logic encoding, and a qualitative-semantics SMT encoding.
The tableau supports unsatisfiable-core extraction for debugging, and can synthesize witness signals for satisfiable formulas by extracting them from the tableau.
To foster reproducibility, we also release an extensive STL benchmark suite~\cite{zenodo}.


\stlsat{} determines in under one second that the~$\varphi_1$--$\varphi_3$ irrigation requirements are mutually inconsistent and the tableau engine also returns the subset of constraints $\set{\varphi_1, \varphi_2}$ responsible for the conflict as an \emph{unsatisfiable core}, thus providing immediate debugging support. 

We conduct an extensive empirical evaluation on the released benchmark set of STL formulas, as well as its close relative MLTL coming from related works, comparing each \stlsat{} engine against state-of-the-art STL solvers and against a dedicated MLTL satisfiability checker~\citep{mltlsat}.
The results show that no single technique dominates across all instances; however, the portfolio approach combining the three \stlsat{} solving engines achieves the best overall performance in terms of runtime and scalability.

\noindent In summary, our contributions are:
\begin{enumerate}[label=\arabic*.,  nosep]
    \item a sound and complete tree-shaped tableau algorithm with a novel set of rules;
    \item \stlsat{}, a high-performance Rust portfolio solver with tableau, FOL, and SMT engines;
    \item a novel unsat-core extraction method for the tableau engine;
    \item an open benchmark suite covering STL formulas; and
    \item a thorough experimental study demonstrating superior performance over existing tools.
\end{enumerate}

The remainder of the paper is organized as follows.
Section~\ref{sec:related-work} surveys related work on temporal logic satisfiability.
Section~\ref{sec:background} recalls the syntax and semantics of STL and MLTL.
Section~\ref{sec:tableau} details the flaws in the existing tableau \cite{MelaniBC25} and presents our corrected tableau algorithm together with its soundness and completeness proofs.
Section~\ref{sec:tool} details the implementation of \stlsat{}'s three solving engines.
Section~\ref{sec:experiments} describes our experimental methodology and reports the results.
Finally, Section~\ref{sec:conclusions} draws conclusions and outlines future research directions.

\section{Related Work}
\label{sec:related-work}

Originally proposed for CPS monitoring and testing \citep{MalerN04}, STL has primarily been studied for runtime verification \citep{BartocciDDFMNS18}, while consistency checking, redundancy removal, and implication analysis remain less explored.

\noindent \textbf{Validity, vacuity, and redundancy.}
These issues were firstly addressed by \citet{Dokhanchi}, who reduced validity, vacuity, and redundancy questions for STL to satisfiability problems in Metric Interval Temporal Logic (MITL) \citep{AlurFH96}. Their approach leverages existing MITL solvers \citep{BersaniRP13} to obtain counterexamples or proofs of equivalence.  While effective, this pipeline incurs the overhead of translating STL into a different logic and does not directly support the extraction of minimal conflicting subsets.

\noindent \textbf{Debugging tools.}
In the same spirit, \citet{RoehmHM17} proposed \textsc{STLInspector}, a debugging aid that generates concrete witness signals by mutating candidate models and querying an SMT encoding of the STL formula.  The tool is useful for pinpointing why a specification holds or fails on a given trace, but it does not provide a systematic analysis of the logical relationships among multiple requirements.

\noindent \textbf{Satisfiability/Model checking.}
The close relationship between satisfiability and bounded model checking has motivated several encodings of STL into first-order logic (FOL) or Satisfiability Modulo Theories (SMT).  \citet{BaeLee19} presented a refutationally complete SMT encoding for bounded model checking of continuous-time STL against hybrid system models; an improved variant was later shown to be suitable for satisfiability checking \citep{LeeYB21}.  For discrete-time settings, \citet{RamanDMMSS14} introduced an SMT formulation aimed at model predictive control that can also serve as a decision procedure for STL satisfiability \citep{MelaniBC25}.

\noindent \textbf{Tableau-based approaches.}
The only dedicated tableau method for bounded discrete-time STL prior to our work is the one-pass tree-shaped construction proposed by \citet{MelaniBC25}.
They proved that STL satisfiability is \textsc{EXPSPACE}-complete and implemented their algorithm, inspired by tableau techniques for other temporal logics \citep{Reynolds16,GeattiGMR21}, in a prototype Python tool.
In the course of our investigation we discovered a soundness flaw in their original jump optimization; we therefore revisited the mathematical definition, and proved both soundness and completeness for the corrected rule.

\noindent \textbf{Related temporal logics.}
Several works have studied satisfiability for logics closely related to STL.  \citet{BersaniRP15,BersaniRP16} offered SMT encodings for MITL and Quantified Temporal Logic (QTL) \citep{HirshfeldR05}, but these frameworks admit only Boolean signals, limiting their applicability to the quantitative nature of STL.  Mission-time Linear Temporal Logic (MLTL) \citep{ReinbacherRS14} shares many syntactic constructs with STL; a systematic comparison of encodings for MLTL—including translations to LTL on finite traces, SMV specifications, and FOL—was carried out by \citet{LiVR22}, who identified the FOL encoding as the best performing.  Our tool incorporates an improved version of this FOL translation and demonstrates better empirical performance on both STL and MLTL benchmarks.

\noindent \textbf{Unsatisfiable-core extraction.}
The extraction of minimal unsatisfiable cores has been extensively investigated for propositional SAT and SMT solving \citep{CimattiGS07}.  Within temporal logics, similar techniques have been explored mainly for Linear Temporal Logic (LTL) \citep{Schuppan16,GeattiGMV24}, but to the best of our knowledge no existing framework provides unsat-core extraction for STL formulas.  By integrating a dedicated core extraction mechanism into our tableau engine we fill this gap and enable direct debugging of inconsistent requirement sets.

\section{Background}
\label{sec:background}

\paragraph{Signal Temporal Logic (STL).}
We define bounded discrete-time STL~\cite{MalerN04} on the temporal domain of natural numbers $\mathbb{N}$.
A \emph{signal} is a function $w : \mathbb{N} \rightarrow \mathbb{R}^n$
that associates time instants with valuations of a finite set of signal variables $S = \{ x_1, \dots, x_n \}$.
We define the projection $w_R$ of a signal $w$ on a set $R \subseteq S$ in the usual way.
STL formulas have the following syntax:
\[
\varphi := \top \mid f(R) \bowtie k \mid \neg \varphi \mid \varphi_1 \lor \varphi_2
  \mid \until{I}{\varphi_1}{\varphi_2}
\]
where $\mathord{\bowtie} \in \{\mathord{<}, \mathord{=}, \mathord{>}\}$,
$R \subseteq S$ is a set of signal variables,
$f : \mathbb{R}^{|R|} \rightarrow \mathbb{R}$ is a linear function of signal values,
$k \in \mathbb{Q}$ is a rational constant,
and $I = [a,b]$ with $a, b \in \mathbb{N}$ and $a \leq b$.

We define the semantics on a signal $w$ and a time instant $t \in \mathbb{N}$ as follows:
\begin{align*}
    &(w,t) \models \top && \text{(always true)} \\
    &(w,t) \models f(R) \bowtie k     & \text{iff } & f(w_R(t)) \bowtie k, \text{ with } \mathord{\bowtie} \in \{\mathord{<}, \mathord{=}, \mathord{>}\} \\
    &(w,t) \models \neg \varphi & \text{iff } & (w,t) \not\models  \varphi \\
    &(w,t) \models \varphi_1 \lor \varphi_2 & \text{iff } & (w,t) \models \varphi_1 \text{ or } (w,t) \models \varphi_2 \\
    &(w,t) \models \until{[a,b]}{\varphi_1}{\varphi_2}  & \text{iff } &
    \exists t' \in [t+a, t+b] : (w, t') \models \varphi_2
    \text{ and } \forall t'' \in [t, t'] : (w,t'') \models \varphi_1
\end{align*}
Boolean signals can be encoded by defining a real variable $x_p$
for each \emph{atomic proposition} $p$, and asserting $p$ as $x_p = 1$ and $\neg p$ as $x_p \neq 1$.

We also use propositional operators ($\land$, $\implies$),
and temporal operators \emph{eventually} (or \emph{finally}) $\even{I}$ and \emph{always} (or \emph{globally}) $\glob{I}$ as shortcuts for $\even{I} \varphi \equiv \until{I}{\top}{\varphi}$ and $\glob{I} \varphi \equiv \neg \even{I} \neg \varphi$.
We also employ the \emph{release} operator $\releaseop$, defined as the dual of \emph{until} under negation, i.e., $\release{I}{\varphi_1}{\varphi_2} \equiv \neg (\until{I}{\neg \varphi_1}{\neg \varphi_2})$.
Given $\until{I}{\varphi_1}{\varphi_2}$ and $\release{I}{\varphi_2}{\varphi_1}$, we refer to $\varphi_1$ as the \emph{invariant}, and to $\varphi_2$ as the \emph{target}.

\paragraph{Mission-Time Linear Temporal Logic (MLTL).}
MLTL~\cite{ReinbacherRS14} can be seen as
a fragment of bounded discrete-time STL that only allows for Boolean signals.
It also differs from STL in the semantics of the until operator, which we denote as strict until ($\suntilop{}$) for MLTL:
\[
(w,t) \models \suntil{[a,b]}{\varphi_1}{\varphi_2}  \text{ iff }
    \exists t' \in [t+a, t+b] : (w, t') \models \varphi_2
    \text{ and } \forall t'' \in [t+a, t'-1] : (w,t'') \models \varphi_1
\]
The strict release operator is defined as
$\srelease{I}{\varphi_1}{\varphi_2} \equiv \neg (\suntil{I}{\neg\varphi_1}{\neg\varphi_2})$.

Indeed, MLTL's until operator is expressible in terms of STL's, and \emph{vice versa}:
\begin{align*}
    \suntil{[a,b]}{\varphi_1}{\varphi_2} &\equiv
    \even{[a,a]} \big(\until{[0, b-a]}{(\varphi_1 \lor \varphi_2)}{\varphi_2}\big) &
    \srelease{[a,b]}{\varphi_1}{\varphi_2} &\equiv
    \even{[a,a]} \big(\release{[0, b-a]}{(\varphi_1 \land \varphi_2)}{\varphi_2}\big) \\
    \until{[a,b]}{\varphi_1}{\varphi_2} &\equiv \glob{[0,a]} \varphi_1 \land \suntil{[a,b]}{\varphi_1}{(\varphi_1 \land \varphi_2)} &
    \release{[a,b]}{\varphi_1}{\varphi_2} &\equiv \even{[0,a]} \varphi_1 \lor \suntil{[a,b]}{\varphi_2}{\varphi_1} \lor \glob{[a,b]} \varphi_2
\end{align*}
The equivalences in the second row, together with common propositional and temporal logic equivalences, allow for rewriting all STL formulas in \emph{strict normal form}, i.e.,
by using only strict operators in positive form, with negation only applied to atomic propositions.

\section{Tree-Shaped Tableau}
\label{sec:tableau}

Building upon the tree-shaped tableau introduced in~\cite{MelaniBC25}, we propose an alternative definition of the \textsf{JUMP} rule, an optimization that skips sequences of redundant time steps during tableau exploration while fixing the soundness and completeness problems of the original rule.

We begin by reviewing the original \emph{basic} tableau, i.e., without the \textsf{JUMP} rule, whose soundness and completeness results still hold.
We then introduce the new \textsf{JUMP} rule and, finally, present its correctness theorems.

\subsection{Basic Tableau}

A tableau is a tree in which each node $u$ is labeled with a set of formula occurrences $\Gamma(u)$ required to hold conjunctively, and a time counter $t(u) \in \mathbb{N}$. 
The root $u_0$ is labeled with the formula $\phi$ to be checked, assumed to be in strict normal form.
Besides a formula, elements of $\Gamma(u)$ may carry auxiliary metadata recording the temporal formula occurrence that generated them---i.e., their \emph{parent}.
We write $\parof{\varphi}{\psi}$ for the occurrence whose underlying formula is $\varphi$ and whose parent occurrence is $\psi$, and $\parof{\varphi}{\bot}$ when no parent is present (e.g., for operators in the input formula).
For instance, in $\parof{a}{\glob{I}a}$, the occurrence $\glob{I}a$ is the parent of $a$, meaning that unfolding it generated that occurrence of $a$.
Parent annotations are only needed by the \textsf{JUMP} rule, hence we omit them whenever not strictly relevant, and we identify occurrences in $\Gamma(u)$ by their underlying formula.

\begin{table}
    \centering
    \caption{
    Expansion rules.
        A node $u$ with $\phi \in \Gamma(u)$ that satisfies the condition is expanded into one child $u_p$ or two children $u_s$ and $u_p$ such that $t(u_s) = t(u_p) = t(u)$, $\Gamma(u_s) = (\Gamma(u) \setminus \{\phi\}) \cup \Gamma_\phi(u_s)$ and $\Gamma(u_p) = (\Gamma(u) \setminus \{\phi\}) \cup \Gamma_\phi(u_p)$. 
    }
    \label{tab:expansion-rules}
    \footnotesize
    \begin{tabular}{c l l l l}
        \toprule
        Cat. & $\phi \in \Gamma(u)$ & condition & $\Gamma_\phi(u_s)$ & $\Gamma_\phi(u_p)$ \\
        
        \midrule
    
        \multirow{2}{*}{Prop.}
            & $\parof{\varphi_1 \lor \varphi_2}{\psi}$
            & 
            & $\parof{\varphi_1}{\psi}$
            & $\parof{\varphi_2}{\psi}$ \\
        
            & $\parof{\varphi_1 \land \varphi_2}{\psi}$
            &
            & $\parof{\varphi_1}{\psi},\ \parof{\varphi_2}{\psi}$
            & \\
    
        \midrule

        \multirow{8}{*}{\rotatebox[origin=c]{90}{Temporal}}
            & \multirow{2}{*}{$\parof{\even{[a,b]}\varphi}{\psi}$}
            & $a \le t(u) < b$
            & $\parof{\exp^{t(u)}(\varphi)}{\bot}$
            & $\parof{\meven{[a,b]}\varphi}{\psi}$ \\
            
            &
            & $t(u)=b$
            & $\parof{\exp^{t(u)}(\varphi)}{\bot}$
            & \\
    
        \cmidrule(lr){2-5}

            & \multirow{2}{*}{$\parof{\glob{[a,b]}\varphi}{\psi}$}
            & $a \le t(u) < b$
            & 
            & $\parof{\mglob{[a,b]}\varphi}{\psi},\ \parof{\exp^{t(u)}(\varphi)}{\phi}$ \\
        
            &
            & $t(u)=b$
            & 
            & $\parof{\exp^{t(u)}(\varphi)}{\phi}$ \\
    
        \cmidrule(lr){2-5}
    
            & \multirow{2}{*}{$\parof{\suntil{[a,b]}{\varphi_1}{\varphi_2}}{\psi}$}
            & $a \le t(u) < b$
            & $\parof{\exp^{t(u)}(\varphi_2)}{\bot}$
            & $\parof{\msuntil{[a,b]}{\varphi_1}{\varphi_2}}{\psi},\ \parof{\exp^{t(u)}(\varphi_1)}{\phi}$ \\
        
            &
            & $t(u)=b$
            & $\parof{\exp^{t(u)}(\varphi_2)}{\bot}$
            & \\
    
        \cmidrule(lr){2-5}
    
            & \multirow{2}{*}{$\parof{\srelease{[a,b]}{\varphi_1}{\varphi_2}}{\psi}$}
            & $a \le t(u) < b$
            & $\parof{\exp^{t(u)}(\varphi_1)}{\bot},\ \parof{\exp^{t(u)}(\varphi_2)}{\bot}$
            & $\parof{\msrelease{[a,b]}{\varphi_1}{\varphi_2}}{\psi},\ \parof{\exp^{t(u)}(\varphi_2)}{\phi}$ \\
        
            &
            & $t(u)=b$
            & 
            & $\parof{\exp^{t(u)}(\varphi_2)}{\phi}$ \\
    
        \bottomrule
    \end{tabular}
\end{table}

The tableau is built incrementally by applying the \emph{expansion} rules in Table~\ref{tab:expansion-rules}, which decompose formulas in a node's label into simpler ones according to their semantics generating child nodes.
Decomposing temporal operators introduces marked operators denoted by a bar (e.g., $\meven{}$), which represent operators whose satisfaction is postponed to the next time step. Furthermore, temporal operator decomposition employs a temporal expansion function $\exp^t$ mapping formulas inside nested temporal operators to their realizations when unfolded at time $t$.

which updates intervals when nested temporal operators are unfolded:
\begin{gather*}
    \begin{alignedat}{3}
        &\exp^t(\top) = \top &\qquad
        &\exp^t(f(R) \bowtie k) = f(R) \bowtie k &\qquad
        &\exp^t(\neg \varphi_1) = \neg \exp^t(\varphi_1)
    \end{alignedat} \\
    \begin{aligned}
        &\exp^t(\varphi_1 \circ \varphi_2) = \exp^t(\varphi_1) \circ \exp^t(\varphi_2) &\qquad& \text{with } \circ \in \set{\land, \lor} \\
        &\exp^t(\anyunop{[a,b]}{\varphi}) = \anyunop{[a+t,b+t]}{\varphi} &\qquad& \text{with } \mathsf{A} \in \set{\evenop, \globop} \\
        &\exp^t(\anybinop{[a,b]}{\varphi_1}{\varphi_2}) = \anybinop{[a+t,b+t]}{\varphi_1}{\varphi_2} &\qquad& \text{with } \mathsf{B} \in \set{\suntilop, \sreleaseop}.
    \end{aligned}
\end{gather*}

During tableau construction, nodes are checked for local inconsistencies. Formally, a node $u$ is rejected (and closed) if either $\neg \top \in \Gamma(u)$ or the set of propositions $\set{f(R) \bowtie k \in \Gamma(u)} \cup \set{\lnot(f(R) \bowtie k) \in \Gamma(u)}$ is inconsistent.
Rejected nodes are leaves.

When no further decomposition rule is applicable to a node, i.e., all temporal operators in its label are marked or not active, the node is said to be \emph{poised}.
In this case, the tableau advances time through the \textsf{STEP} rule:
\begin{description}
    \item[\textsf{STEP}]
    Let $u$ be a poised node. If $\Gamma(u)$ contains any temporal operators, then $u$ has one child $u'$ with $t(u')=t(u)+1$, and
    \begin{align*}
    \Gamma(u') =
    \ & \set{\anyunop{I}{\varphi} \in \Gamma(u) \mid \mathsf{A} \in \set{\evenop, \globop}}
    \cup \set{\anybinop{I}{\varphi_1}{\varphi_2} \in \Gamma(u) \mid \mathsf{B} \in \set{\suntilop, \sreleaseop}} \\
    &\cup \set{\anyunop{[a,b]}{\varphi} \mid \stlunop{\marked{\mathsf{A}}}{[a,b]}{\varphi} \in \Gamma(u) \land \mathsf{A} \in \set{\evenop, \globop} \land t(u) < b} \\
    &\cup
    \set{\anybinop{[a,b]}{\varphi_1}{\varphi_2} \mid \stlbinop{\marked{\mathsf{B}}}{[a,b]}{\varphi_1}{\varphi_2} \in \Gamma(u) \land \mathsf{B} \in \set{\suntilop, \sreleaseop} \land t(u) < b}
    \end{align*}
\end{description}

If $u$ is poised, not rejected, and its \textsf{STEP} successor label is empty,
then the corresponding branch is accepted and witnesses the satisfiability
of the input formula.
If the tableau frontier consists of only rejected nodes, the input formula is unsatisfiable.

Fig.~\ref{fig:tableau-example-basic} shows an example of the basic tableau for formula $\glob{[0,4]}a \land \suntil{[4,7]}{b}{\neg a}$.
The tableau starts with the expansion rule for the $\land$ operator, and proceeds with the one for the $\globop$ operator, which creates a marked instance of $\glob{[0,4]}a$ and extracts its argument $a$.
The interval of the $\suntilop$ operator starts at 4, so no more expansion rules can be applied and $u_2$ is \emph{poised}, and the \textsf{STEP} rule advances time by creating node $u_3$.
This process goes on until time 4, creating many nodes identical to $u_1$ and $u_2$, which we omit from the figure.
At time 4, the $\suntilop$ operator is expanded, creating two nodes $u_s = u_6$ and $u_p = u_7$.
In $u_6$, we try to \emph{satisfy} it by extracting $\neg a$, which conflicts with $a$ extracted by the $\globop$ operator and leads to a rejected node.
In $u_7$, we \emph{postpone} the satisfaction of the $\suntilop$ operator by marking it and extracting $b$, which does not conflict with any other proposition.
Since $\glob{[0,4]}a$ ends at time 4, the step rule does not propagate it to $u_8$.
Thus, the $\suntilop$ can finally be satisfied in $u_9$.
Since $u_9$ has no temporal operators and no conflicts between atomic propositions, it is accepted, witnessing the satisfiability of the input formula.

Soundness and completeness of the basic tableau were established in~\cite{MelaniBC25}:
\begin{theorem}[{\cite[Theorem 4.4 and Lemma 3]{MelaniBC25}}]
    A formula $\varphi$ is satisfiable iff the basic tableau rooted in it has an accepting branch.
\end{theorem}

\begin{figure}
    \begin{minipage}{0.38\textwidth}
        \footnotesize
        \begin{forest}
  for tree={
    myleaf/.style={label=below:{#1}},
    l sep=4mm,
    s sep=4mm
  },
  [{$u_0:\ \glob{[0,4]}a \land \suntil{[4,7]}{b}{\neg a}\mid \mathbf{0}$}
    [{$u_1:\ \glob{[0,4]}a,\ \suntil{[4,7]}{b}{\neg a}\mid \mathbf{0}$},
     edge label={node[midway,left,font=\scriptsize]{$\land$}}
        [{$u_2:\ \mglob{[0,4]}a,\ a,\ \suntil{[4,7]}{b}{\neg a}\mid \mathbf{0}$},
        edge label={node[midway,left,font=\scriptsize\sffamily]{G}}
        [{$u_3:\ \glob{[0,4]}a,\ \suntil{[4,7]}{b}{\neg a}\mid \mathbf{1}$},
            edge label={node[midway,left,font=\scriptsize\sffamily]{STEP}}
            [{$u_4:\ \glob{[0,4]}a,\ \suntil{[4,7]}{b}{\neg a}\mid \mathbf{4}$},
            edge label={node[midway,font=\scriptsize\sffamily]{\dots}},
            no edge
                [{$u_5:\ a,\ \suntil{[4,7]}{b}{\neg a}\mid \mathbf{4}$},
                edge label={node[midway,left,font=\scriptsize\sffamily]{G}}
                [{$u_6:\ a,\ \neg a\mid \mathbf{4}$}, myleaf={\xmark}]
                [{$u_7:\ a,\ \msuntil{[4,7]}{b}{\neg a},\ b\mid \mathbf{4}$},
                    edge label={node[midway,left,xshift=-18pt,font=\scriptsize\sffamily]{$\suntilop$}}
                    [{$u_8:\ \suntil{[4,7]}{b}{\neg a}\mid \mathbf{5}$},
                    edge label={node[midway,left,font=\scriptsize\sffamily]{STEP}}
                        [{$u_9:\ \neg a\mid \mathbf{5}$}, label=left:{\cmark}]
                        [{$u_{10}:\ \msuntil{[4,7]}{b}{\neg a},\ b\mid \mathbf{5}$},
                        edge label={node[midway,left,xshift=-18pt,font=\scriptsize\sffamily]{$\suntilop$}}
                        ]
                    ]
                ]
                ]
            ]
        ]
        ]
    ]
  ]
\end{forest}
    \end{minipage}%
    \begin{minipage}{0.38\textwidth}
        \footnotesize
        \begin{forest}
  for tree={
    myleaf/.style={label=below:{\strut#1}},
    l sep=5mm,
    s sep=5mm
  },
  [{$u_0:\ \glob{[0,5]}\bigl(p\lor\glob{[2,2]}p\bigr)\land\glob{[3,5]}q\mid\mathbf{0}$}
        [{$u_1:\ \mglob{[0,5]}\bigl(p\lor\glob{[2,2]}p\bigr),\ p\lor\glob[{[0,5]}]{[5,5]}p,\ \glob{[3,5]}q\mid\mathbf{3}$},
         edge label={node[midway,font=\scriptsize\sffamily]{\dots}},
         no edge
          [{$u_2:\ \mglob{[0,5]}\bigl(p\lor\glob{[2,2]}p\bigr),\ p\lor\glob[{[0,5]}]{[5,5]}p,\ \mglob{[3,5]}q,\ q\mid\mathbf{3}$},
            edge label={node[midway,left,font=\scriptsize\sffamily]{G, G}}
            [{$u_3:\ \dots,\ p,\ q\mid\mathbf{3}$},
              edge label={node[midway,left,xshift=22pt,font=\scriptsize\sffamily]{$\lor$}},
              myleaf={\xmark}
            ]
            [{$u_4:\ \mglob{[0,5]}\bigl(p\lor\glob{[2,2]}p\bigr),\ \mglob{[3,5]}q,\ \glob[{[0,5]}]{[5,5]}p,\ q\mid\mathbf{3}$},
              edge label={node[midway,right,xshift=-22pt,font=\scriptsize\sffamily]{$\lor$}}
              [{$u_5:\ \glob{[0,5]}\bigl(p\lor\glob{[2,2]}p\bigr),\ \glob{[3,5]}q,\ \glob[{[0,5]}]{[7,7]}p\mid\mathbf{5}$},
                edge label={node[midway,left,font=\scriptsize\sffamily]{JUMP}}
                [{$u_6:\ p\lor\glob[{[0,5]}]{[7,7]}p,\ q,\ \glob[{[0,5]}]{[7,7]}p\mid\mathbf{5}$},
                  edge label={node[midway,left,font=\scriptsize\sffamily]{G, G}}
                  [{$u_7:\ \dots,\ p,\ q\mid\mathbf{5}$},
                    edge label={node[midway,left,xshift=18pt,font=\scriptsize\sffamily]{$\lor$}},
                    myleaf={\xmark}
                  ]
                  [{$u_8:\ q,\ \glob[{[0,5]}]{[7,7]}p\mid\mathbf{5}$},
                    edge label={node[midway,right,xshift=-18pt,font=\scriptsize\sffamily]{$\lor$}}
                    [{$u_{9}:\ \glob[{[0,5]}]{[7,7]}p\mid\mathbf{7}$},
                      edge label={node[midway,left,font=\scriptsize\sffamily]{JUMP}}
                      [{$u_{10}:\ p\mid\mathbf{7}$}, label=left:{\cmark},
                        edge label={node[midway,left,font=\scriptsize\sffamily]{G}},
                      ]
                    ]
                  ]
                ]
              ]
            ]
          ]
        ]
  ]
\end{forest}
    \end{minipage}
    \begin{minipage}{0.5\textwidth}
        \captionof{figure}{Basic tableau for $\glob{[0,4]}a \land \suntil{[4,7]}{b}{\neg a}$.}
        \label{fig:tableau-example-basic}
    \end{minipage}%
    \begin{minipage}{0.5\textwidth}
        \captionof{figure}{\stltree{} tableau for $\varphi_S$.}
        \label{fig:tableau-example-unsound}
    \end{minipage}
\end{figure}

\subsection{The \texorpdfstring{\oldjump{}}{JUMPO} rule and its unsoundness}
\label{sec:unsound-jump}

As an optimization to the basic tableau, \citet{MelaniBC25} introduced the \oldjump{} rule to skip repetitive parts of the tableau by jumping directly to a later point in time.
For instance, for formula $\glob{[0,20]}{x > 0} \land \even{[0,5]}{x > 5}$, the tableau need not explore every time instant, as checking consistency only at interval bounds 0, 5, and 20 is sufficient to determine satisfiability.
Formulas with nested operators are trickier: the tableau by \cite{MelaniBC25} waits until the end of the time interval of the first extracted instance of the inner operator.
For instance, with formula $\glob{[5,20]}{\even{[0,5]} x > 0}$, the tableau proceeds with the \textsf{STEP} rule until time $5 + 5 = 10$, after which it may jump to time 20.

We found out, however, that this precaution is not enough to preserve the soundness of the tableau.
As we mentioned in \S\ref{sec:intro}, formula $\varphi_S$ witnesses this flaw.
Fig.~\ref{fig:tableau-example-unsound} shows part of the tableau rooted in $\varphi_S$ that uses the \oldjump{} rule.
The figure omits the prefix up to time 3, in which the first conjunct is satisfied by choosing $p$ and the \textsf{STEP} rule is applied.
At time 3, expanding the two outer $\globop$ operators produces $q$ and the disjunction $p \lor \glob[{[0,5]}]{[5,5]}p$
(here, the superscript $[0,5]$ indicates that the formula is extracted from an operator with interval $[0,5]$).
The branch that chooses $p$ is locally inconsistent with $q$ and is rejected.
The other branch chooses $\glob[{[0,5]}]{[5,5]}p$.

Since its current time is greater than $0+2=2$, the poised node $u_4$ satisfies the conditions of the \oldjump{} rule, which jumps from time 3 to time 5.
Since the parent of $\glob[{[0,5]}]{[5,5]}p$ is still active at time 3, the rule treats this formula as a derived obligation and shifts both interval bounds by the jump length 2.
Consequently, $u_6$ contains $\glob[{[0,5]}]{[7,7]}p$ instead of the required $\glob[{[0,5]}]{[5,5]}p$.
At time 5, the first conjunct can again select its nested $\globop$ disjunct, which yields the same obligation at time 7; after $q$ expires, the tableau satisfies $p$ at time 7 and accepts the branch.
The branch therefore never checks the inconsistent pair $p,q$ at time 5.


The \oldjump{} rule compromises both soundness and completeness of the tableau.
We describe this flaw in more detail in
\appendixreference{app:jump-flaws}{A}.

\subsection{A new, correct \textsf{JUMP} rule}
\label{sec:jump}

To speed up tableau construction, we introduce the \textsf{JUMP} rule, which can be used to compress a sequence of \textsf{STEP} rule applications.
Intuitively, a jump from a poised node $u$ to a node $u'$ has the same effect as postponing all active temporal operators until time $t(u')$ by always choosing the expansion rule $\Gamma_{\phi}(u_p)$.

To describe the \textsf{JUMP} rule, we introduce new notation, then discuss the jump-size computation, and finally describe the resulting node. 
To simplify notation, we omit the $\globop$ and $\evenop$ operators, which we replace by the standard equivalences $\glob{I}\varphi \equiv \srelease{I}{\neg\top}{\varphi}$ and $\even{I}\varphi \equiv \suntil{I}{\top}{\varphi}$.

The \textsf{JUMP} rule needs parent annotations, which are created by expansion rules.
When the \textsf{JUMP} rule is enabled, the \textsf{STEP} rule propagates parent annotations on all formulas that it inserts in the new node.
They are accessed through the \emph{parent-active} predicate $\pactive$, which we define over a formula occurrence and a node as follows:
\begin{align*}
    \pactive(\phi,u)
    \quad\text{iff}\quad
    \phi = \parof{\varphi}{\psi}
    \text{ for some } \varphi, \psi \in \Gamma(u)
    \text{ s.t. } \psi \neq \bot
    \text{ and }
    \psi \in \Gamma(u).
\end{align*}

\paragraph{Proposition Validity Intervals.}

We define the proposition validity intervals $\tinv(\varphi)$ of a formula $\varphi$ as the time intervals in which its atomic propositions may appear, allowing us to identify potential conflicts. Formally:
\begin{gather*}
    \begin{aligned}
        &\tinv(f(R) \bowtie k) = \tinv(\top) = \set{(0,0)} &\qquad
        &\tinv(\neg \varphi) = \tinv(\varphi) \\
        &\tinv(\varphi_1 \circ \varphi_2) = \tinv(\varphi_1) \cup \tinv(\varphi_2) &\qquad& \text{with } \circ \in \set{\land, \lor} \\
    \end{aligned} \\
    \begin{aligned}
        \tinv(\suntil{[a,b]}{\varphi_1}{\varphi_2})
        =
        \tinv(\srelease{[a,b]}{\varphi_1}{\varphi_2})
        &=
        \set{(l+a, r+b-1) \mid (l,r) \in \tinv(\varphi_1),\ a < b} \\
        &\quad\cup
        \set{(l+a, r+b) \mid (l,r) \in \tinv(\varphi_2)}
    \end{aligned}
\end{gather*}

To distinguish different occurrences of the same proposition, we associate a unique identifier to every occurrence.
For instance, in $\varphi = \suntil{[0,5]}{p}{q} \land \suntil{[0,10]}{p}{r}$, the two occurrences of $p$ are treated as distinct because they are governed by different temporal constraints.
Formally, the formula is viewed as a syntax tree whose leaves carry unique identifiers.
Two propositions are considered the same only if they share the same syntactic structure and identifier.
Identifiers are preserved by expansion rules and by the temporal expansion operator $\exp$.
Such identifiers are used only for validity intervals: for local consistency checking, propositions are compared solely based on their syntactic structure.

Function $\tinv$ is defined so that each occurrence of an atomic proposition $\rho$ in a formula contributes exactly one pair $(l,r)$. 
Hence, we may unambiguously refer to a pair $i = (l,r) \in \tinv(\varphi)$ deriving from $\rho$, meaning the specific occurrence identified via its identifier rather than merely its syntactic structure.
We denote as $\id{i}$ such an identifier.

The following lemma, whose proof is in
\appendixreference{app:proofs}{B}, relates validity intervals computed from a formula $\varphi$ to the time instants at which occurrences of a proposition produced by decompositions of $\varphi$ may appear in the tableau:
\begin{lemma}
    \label{lemma:te-limit}
    Let $u$ be a node of the basic tableau, $\varphi \in \Gamma(u)$ be a formula in the label of the node, and $\rho$ be an atomic proposition appearing in $\varphi$.
    Let $(l, r) \in \tinv(\varphi)$ be the proposition validity interval deriving from $\rho$.
    Then, in the subtree rooted at $u$, any occurrence of $\rho$ derived from the expansion of $\varphi$ can only appear in nodes $v$ such that $l \leq t(v) \leq r$ if $\rho$ is nested within a temporal operator in $\varphi$, and such that $t(v) = t(u)$ otherwise.
\end{lemma}


\subsubsection{Jump Size Calculation}
\label{sec:jump-size-calculation}

To guarantee that the \textsf{JUMP} rule preserves the correctness of the algorithm, we compute the jump size $k^{\ast}$ as the minimum of two bounds.

\paragraph{Bounds Limit ($k$)} 

The first limit ensures that the jump does not cross a time instant at which an operator becomes active or expires.
It is defined as
\[
    k(u) = \min\{t\in K(u)\mid t>t(u)\} - t(u),
\]
where the set $K(u)$ collects bounds of temporal operators in $u$:
\[
    K(u) = \set{a, b\mid \anybinop{[a,b]}{\varphi_1}{\varphi_2}\in\Gamma(u),\ \mathsf{B}\in\set{\suntilop,\marked{\suntilop},\sreleaseop,\marked{\sreleaseop}}}
\]

We illustrate the computation of the bounds on node $u_2$ from the tableau in Fig.~\ref{fig:tableau-example-basic}.
We have $K(u_2) = \set{0, 4, 7}$, and thus $k(u_2) = 4$.

\paragraph{Soundness Limit ($k^{\ast}_{\text{sound}}$)}

Consider a temporal operator $\varphi = \msuntil{[a,b]}{\varphi_1}{\varphi_2}$ or $\varphi = \msrelease{[a,b]}{\varphi_2}{\varphi_1}$ in a poised node $u$.
The basic tableau would emit the invariant part $\varphi_1$ at every
intermediate time instant $t(u)+j$, with $1 \leq j < k$, before reaching the
jump destination $t(u)+k$.
To guarantee the soundness of the approach, we need to be sure that such invariants do not generate conflicts with other formulas in the node. 

We first collect the validity intervals of atomic propositions that may be generated by the invariant part of active postponed temporal operators:
\begin{align*}
    N(u) = \bigcup\limits_{\varphi \in \Gamma(u)} \set{(t(u) + l, t(u) + r) \mid (l, r) \in \tinv(\varphi_1), \varphi = \msuntil{[a,b]}{\varphi_1}{\varphi_2} \lor \varphi = \msrelease{[a,b]}{\varphi_2}{\varphi_1}}.
\end{align*}
An element $(l_n,r_n)\in N(u)$ describes the validity interval of an invariant emitted at time instant $t(u)$ by an active temporal operator $\varphi$, i.e., $\exp^{t(u)}(\varphi_1)$.
If the same invariant is emitted at a skipped time $t(u)+j$, its interval is shifted to $(l_n+j,r_n+j)$.

We then collect the validity intervals of atomic propositions that may be
generated by temporal operators in $\Gamma(u)$ whose parent is not active:
\begin{align*}
    O(u) = \bigcup\limits_{\varphi \in \Gamma(u)} \set{(l,r) \mid (l,r) \in \tinv(\varphi), \varphi = \anybinop{[a,b]}{\varphi_1}{\varphi_2}, \mathsf{B} \in \set{\suntilop, \marked{\suntilop}, \sreleaseop, \marked{\sreleaseop}}, \lnot \pactive(\varphi, u)} 
\end{align*}

If an interval in $N(u)$ intersects one in $O(u)$, we cannot exclude that the skipped invariant conflicts with other formulas.
Since we do not know which propositions will eventually form an inconsistent core, we perform a \textsf{STEP} as a safe choice:
the \textsf{JUMP} rule is disabled whenever
\begin{equation}
    \label{eq:sound-condition}
    \exists (n, o) \in N(u) \times O(u) \text{ s.t. } \id{n} \neq \id{o} \land n \cap o \neq \emptyset
\end{equation}

If no intersection exists, for every
$n=(l_n,r_n)\in N(u)$ and $o=(l_o,r_o)\in O(u)$, either $r_n<l_o$ or $r_o<l_n$.
The case $r_o < l_n$ is always safe: delaying the emitted invariants moves the interval to the future, so other proposition intervals will never intersect with it.
In case $r_n < l_o$ instead, shifting the invariant to the future might cause the intersection of the intervals.
We define the sound jump limit as the earliest distance at which propositions may appear together:
\[
    k^{\ast}_{\text{sound}}(u) = \min\big(\set{l_o - r_n \mid (n, o) \in N(u) \times O(u) \land \id{n} \neq \id{o} \land r_n < l_o} \cup \set{+\infty}\big)
\]

Let us evaluate the soundness bounds on node $u_2$ in Fig.~\ref{fig:tableau-example-basic}:
we have $N(u_2) = \set{(0,0)}$ and $O(u_2) = \set{(0,4), (4,6), (4,7)}$, hence $k^{\ast}_{\text{sound}}(u_2) = 4$.

\paragraph{Completeness Guard}

Consider a temporal operator $\varphi = \msuntil{[a,b]}{\varphi_1}{\varphi_2}$ or $\varphi = \msrelease{[a,b]}{\varphi_2}{\varphi_1}$ in a poised node $u$.
The \textsf{JUMP} rule reconstructs the branch in which this temporal obligation is postponed until the jump destination $t(u) + k$. 
In the basic tableau, an accepting branch may satisfy the temporal obligation (i.e., select the $\Gamma_{\varphi}(u_s)$ expansion rule) at some intermediate time $t(u)+j$, with $1 \leq j < k$.
If time step $t(u)+j$ is skipped, the \textsf{JUMP} rule may reach a node rejected due to the postponed obligation, even though the basic tableau had an accepting branch rooted at $t(u)+j$.

We first collect the validity intervals of atomic propositions that may be generated by the target (i.e., $\varphi_2$) part of active temporal operators:
\begin{align*}
    M(u) = \bigcup\limits_{\varphi \in \Gamma(u)} \set{(t(u) + l, t(u) + r) \mid (l, r) \in \tinv(\varphi_2), \varphi = \msuntil{[a,b]}{\varphi_1}{\varphi_2} \lor \varphi = \msrelease{[a,b]}{\varphi_2}{\varphi_1} }.
\end{align*}
An element $(l_m,r_m)\in M(u)$ describes the validity interval of a proposition generated by the formula emitted at time $t(u)$, namely $\exp^{t(u)}(\varphi_2)$.
If the same target is emitted at some time $t(u)+j$ with $1 \leq j < k$, its interval is shifted to $(l_m+j,r_m+j)$.

We then collect the validity intervals of propositions that may generate conflicts with a target extraction.
These consist of propositions already present in the node and the ones generated by non-expired temporal operators without an active parent:
\[
    S(u) = O(u)
    \cup
    \set{
        (t(u),t(u)) \mid
        \varphi \in \Gamma(u),\
        \varphi = f(R)\bowtie k \lor \varphi = \neg(f(R)\bowtie k),\
        \lnot\pactive(\varphi,u)
    }.
\]

If an interval in $M(u)$ intersects one in $S(u)$, then we cannot determine when the invariant stops conflicting with the target.
Since the extraction might be satisfiable at the next instant, we perform a \textsf{STEP}, and disable the \textsf{JUMP} rule whenever:
\begin{equation}
    \label{eq:complete-condition}
    \exists m \in M(u), \exists s \in S(u) \text{ s.t. } \id{m} \neq \id{s} \land m \cap s \neq \emptyset
\end{equation}

If no such intersection exists, then for every
$m=(l_m,r_m)\in M(u)$ and $s=(l_s,r_s)\in S(u)$, either $r_m<l_s$ or $r_s<l_m$.
Case $r_s < l_m$ is safe, as the potentially conflicting interval is before the target one, and delaying extraction would move the latter to the future.
Otherwise, when $r_m < l_s$, extracting the target at a later time can cause the two intervals to intersect, leading to a possible conflict.
However, if the \textsf{JUMP} rule is enabled, then no two intervals $m$ and $s$ intersect, and all currently active temporal formulas can be satisfied by extracting their targets at the current time.
Thus, there is no need for a jump limit to ensure completeness.

In the running example of Fig.~\ref{fig:tableau-example-basic}, we have $M(u_2) = \set{(0,0)}$ and $S(u_2) = \set{(0,4), (4,7)}$.
Note that proposition $a$ has been extracted from the active operator $\glob{[0,4]}a$,
so it is marked with $\parof{a}{\glob{[0,4]}a}$, and its validity interval is excluded from $S(u_2)$.

\subsubsection{Rule Application}
\label{sec:rule-application}

Let $u$ be a poised node in which \eqref{eq:sound-condition} and \eqref{eq:complete-condition} do \emph{not} hold.
The \textsf{JUMP} rule creates one child $u'$ with $t(u') = t(u) + k^{\ast}(u)$, where the jump size $k^{\ast}(u)$ and node label $\Gamma(u')$ are defined as
\begin{gather*}
    k^{\ast}(u) = \min(k(u), k^{\ast}_{\text{sound}}(u)) \\
\begin{aligned}
    \Gamma(u')= \
    &\set{\parof{\suntil{[a,b]}{\varphi_1}{\varphi_2}}{\psi}\mid \parof{\anybinop{[a,b]}{\varphi_1}{\varphi_2}}{\psi} \in\Gamma(u), \  \mathsf{B}\in\set{\suntilop,\marked{\suntilop}}, \ t(u') \leq b} \\
    &\cup\set{\parof{\srelease{[a,b]}{\varphi_1}{\varphi_2}}{\psi} \mid  \parof{\anybinop{[a,b]}{\varphi_1}{\varphi_2}}{\psi}\in\Gamma(u), \ \mathsf{B}\in\set{\sreleaseop,\marked{\sreleaseop}}, \ t(u') \leq b}
\end{aligned}
\end{gather*}

The label $\Gamma(u')$ omits real-valued constraints and unmarks operators, so that expansion rules can build nodes for the next time instants.
The bound $k^{\ast}_{\text{sound}}(u)$ guarantees that the \textsf{JUMP} rule preserves remaining temporal obligations while avoiding redundant checks.

Finally, we can evaluate the \textsf{JUMP} bounds on node $u_2$ of Fig.~\ref{fig:tableau-example-basic}.
We have $k(u_2) = 4$ and $k^{\ast}_{\text{sound}}(u_2) = 4$;
hence, the \textsf{JUMP} rule can be applied with a jump size of 4, skipping time instants 1 to 3.
With the \textsf{JUMP} rule, node $u_4$ is the (only) child of $u_2$.

\paragraph{Shared-Variable Optimization}
For a set of atomic constraints to be inconsistent, at least two of them have to share a variable.
Moreover, two occurrences of a Boolean variable cannot conflict if they have the same polarity.
From this observation we can derive a simple optimization that allows for applying the \textsf{JUMP} rule more often:
in \eqref{eq:sound-condition} and \eqref{eq:complete-condition} we consider only
pairs of intervals that come either from two opposite occurrences of the same Boolean variable (such as $a$ and $\neg a$)
or from two atomic constraints such that at least one variable occurs in both.

\subsubsection{Correctness}

We prove soundness and completeness of the \textsf{JUMP} rule by showing, respectively, that an accepted branch in the \textsf{JUMP} tableau implies one in the basic tableau, and \emph{vice versa}.

\begin{lemma}
    \label{lemma:jump-sound}
    Let $u$ be a poised node. Let $u'$ be the node obtained by applying the JUMP rule to $u$, or the STEP rule if JUMP conditions are not satisfied. If the subtree rooted at $u'$ contains an accepted branch, then the subtree of the basic tableau rooted at $u$ also contains one.
\end{lemma}

\begin{theorem}[Soundness]
    \label{thm:jump-sound}
    If the tableau built by using the \textsf{JUMP} rule has an accepted branch, then the basic tableau for a formula $\phi$ also has an accepted branch
\end{theorem}

\begin{theorem}[Completeness]
    \label{thm:jump-complete}
    If the basic tableau for a formula $\phi$ has an accepted branch, then the tableau built by using the \textsf{JUMP} rule also has an accepted branch.
\end{theorem}

We report manually written proofs of these results in
\appendixreference{app:proofs}{B}.
Moreover, we have formalized the above claims and their proofs with the Lean 4 proof assistant \cite{deMouraU21},
with the help of GPT 5.6 Sol \cite{openai2026gpt56}.
The Lean formalization is publicly available in \cite{zenodo}.

\section{\stlsat{} Tool}
\label{sec:tool}

\stlsat{}~\cite{stlsat} is a modular tool for satisfiability checking of bounded discrete-time STL formulas.
The full implementation, including source code, benchmarks, and reproducibility scripts, is publicly available\footnote{\url{https://github.com/ZamponiMarco/stlsat/}}.
Given a formula $\varphi$, \stlsat{} determines whether there exists a discrete-time signal $w: \mathbb{N} \to \mathbb{R}^n$ such that $(w, 0) \models \varphi$. 
\stlsat{} implements three different solving engines working over the same formula representation:
\begin{itemize}[nosep]
    \item the tableau-based method described in \S\ref{sec:tableau};
    \item a FOL engine for STL inspired by~\cite{LiVR22};
    \item a quantifier-free encoding inspired by MPC~\cite{RamanDMMSS14} and bounded model checking~\cite{ClarkeBRZ01} methods.
\end{itemize}

\subsection{Tableau Method}

\stlsat{} features an efficient implementation of the tableau methodology described in \S\ref{sec:tableau}. 

It introduces a lightweight syntactic simplification phase before tableau construction. This pre-processing step applies common rewriting rules, including nested conjunction and disjunction flattening, redundant operand removal (e.g., $\varphi \land \top \equiv \varphi$), and merging of temporal operator intervals (e.g., $\glob{[0,3]}\varphi \land \glob{[2,5]}\varphi \equiv \glob{[0,5]}\varphi$). These simplifications reduce formula complexity and, consequently, the depth of tableau expansion without affecting satisfiability.

Furthermore, \stlsat{} replaces the tableau expansion of~\cite{MelaniBC25}, which relies on the program call stack for node exploration, with an explicit stack-based implementation. This design enables finer control over memory allocation and improves performance when exploring deep tableaux.

\stlsat{} performs local consistency checks of atomic propositions by invoking the Z3 solver \cite{Z3}.
We also implemented a custom solver for systems of difference constraints by reducing them to negative-cycle detection in graphs~\cite{cormen2022introduction,armando2004sat,DutertreM06}.
This optimization improves performance by eliminating the external Z3 calls for supported formulas.

\paragraph{Witness Signals}
\stlsat{} is able to extract witness signals for satisfiable STL formulas.
The witness signal is constructed by traversing the accepting branch of the basic tableau and collecting the atomic propositions extracted at each time instant.
The resulting signal is then represented as a piecewise function over the time domain, with values determined by the extracted propositions.

\paragraph{\textbf{Unsat-Core Extraction}}

Whenever a formula is determined to be unsatisfiable, \stlsat{} performs a lightweight unsatisfiable-core extraction to identify the source of inconsistency. 
We represent the input formula $\varphi$ as a conjunction of subformulas after repeatedly applying the $\land$-expansion rule, so that at the root node $u_0$ we have
\(
    \Gamma(u_0) = \set{\varphi_1, \dots, \varphi_m},
\)
and
$\varphi = \bigwedge_{i=1}^m \varphi_i$.
Let $\mathcal{A}$ be the set of atomic predicates $f(R)\bowtie k$ occurring in $\varphi$.
We associate each atomic predicate $a \in \mathcal{A}$ to its originating subformula $\varphi_i$ through a mapping $M: \mathcal{A} \to \set{1, \dots, m}$.

During tableau expansion, every node contains a set of atomic predicates $A(u) \subseteq \mathcal{A}$ extracted from the input formula $\varphi$. When a node is rejected by the underlying theory solver, the solver provides the subset of conflicting atomic predicates responsible for the inconsistency:
\begin{equation*}
    C_u = \set{a_{u,1}, a_{u,2}, \dots, a_{u,k}} \subseteq A(u).
\end{equation*}

When all branches of the tableau are rejected, the local unsatisfiable cores are aggregated into a single set $C = \bigcup_{u\text{ rejected}} C_u$ and mapped back to their originating formulas in the root node, producing a high-level unsatisfiable core $U \subseteq \Gamma(u_0)$ that identifies a subset of subformulas responsible for the inconsistency, with
\(
    U = \set{\varphi_i \mid \exists a \in C : M(a) = i}.
\)
The returned core is unsatisfiable, but not necessarily irreducible or minimal. Its further evaluation and minimization, as well as solver-native core extraction for the remaining engines, are left to future work.

\subsection{First-Order Logic Encoding}
\label{sec:fol-encoding}

An almost out-of-the-box solution for STL satisfiability is to translate a formula $\varphi$ into First-Order Logic (FOL) by recursively replacing each operator with its semantics definition, and then check the resulting formula with a theorem prover or an SMT solver with quantifier support.

\citet{LiVR22} applies this idea to MLTL. They encode signals as functions from non-negative integers (representing time) to Boolean values: for each atomic proposition $p$, they define a function $p : \mathbb{N} \rightarrow \{\top, \bot\}$, such that $p(t) = \top$ iff $p$ holds at time $t$.
Moreover, they define a Boolean function $f$ taking two integers $k$ and $\mathit{len}$, which returns $\top$ iff the FOL translation of $\varphi$ is true at time $k$ for a signal of length $\mathit{len}$.

We implement an optimized STL encoding where signals are integer-to-real functions (for real values) or integer-to-Boolean functions (for Boolean signals).  The translation follows the semantics in \S\ref{sec:background}. 
To reduce formula size, we employ the usual \emph{ad hoc} FOL translations for the $\evenop{}$ and $\globop{}$ operators, rather than translating them to $\untilop{}$.

\subsection{Quantifier-Free SMT Encoding}
\label{sec:smt-encoding}

Closely related to the FOL approach, STL satisfiability can also be encoded as a quantifier-free SMT problem over linear real arithmetic, similarly to SMT encodings used for optimal control~\cite{RamanDMMSS14} and bounded model checking~\cite{ClarkeBRZ01}.

The encoding first computes the formula horizon and creates symbolic variables for each signal variable at every time instant. The encoding recursively unfolds the semantics of STL operators. Boolean operators are translated directly into their SMT counterparts, while temporal operators are expanded over their intervals. The resulting quantifier-free linear real arithmetic formula can be solved by an SMT solver.

\section{Experiments}
\label{sec:experiments}

We compare the three \stlsat{} engines against state-of-the-art STL and MLTL satisfiability tools.
For both logics, we include \stltree{}~\cite{stltree}, the original Python implementation of~\cite{MelaniBC25}; for MLTL, we also compare with the best-performing technique for MLTL satisfiability analyzed by \citet{LiVR22}, namely their FOL encoding which we call MLTLSAT, obtained by running the tool made available by the authors \cite{mltlsat}. MLTLSAT translates MLTL formulas into SMTLIB files, which we feed to Z3 4.15.8 \cite{Z3}.
We ran all experiments on a server equipped with an AMD EPYC 9654 CPU and 755 GB of RAM
running Ubuntu 24.04.
The parallel configuration runs all three engines concurrently and reports the wall-clock time until the first engine terminates, thus representing a portfolio result.
A replication package can be found in \cite{zenodo}.

\subsection{Benchmarks}
\label{sec:benchmarks}

We evaluate \stlsat{} on MLTL and STL benchmark suites, whose descriptive statistics are reported in Table~\ref{tab:benchmark_statistics}.

\begin{table}
    \centering
    \caption{Descriptive statistics of the MLTL and STL benchmark suites. Values show median (95\textsuperscript{th} percentile) for depth, temporal depth and horizon, and mean for branching factor (i.e., number of node children). The last four columns show the frequency (\%) of each operator.}
    \label{tab:benchmark_statistics}
    \footnotesize
    \begin{tabular}{clccccrrrr}
        \toprule
        \multicolumn{2}{l}{\textbf{Bench. suite}} & \textbf{Depth} & \textbf{T. Depth} & \textbf{Horizon} & \textbf{B. F.} & $\evenop$ & $\globop$ & $\untilop$ & $\releaseop$ \\
        \midrule
        \multirow{3}{*}{\rotatebox[origin=c]{90}{MLTL}}
        & NASA-Boeing & 8.0 (11.0) & 1.0 (1.9) & 100000.0 (100001.8) & 1.80 & 12.8 & 87.2 & 0.0 & 0.0 \\
        & Random & 10.0 (20.0) & 6.0 (18.0) & 461.0 (1374.0) & 1.46 & 36.8 & 37.1 & 15.9 & 10.3 \\
        & Random0 & 10.0 (20.0) & 6.0 (18.0) & 316.0 (943.0) & 1.46 & 37.0 & 36.8 & 15.6 & 10.6 \\
        \midrule
        \multirow{2}{*}{\rotatebox[origin=c]{90}{STL}}
        & Random-STL & 10.0 (21.0) & 4.0 (8.0) & 1000.0 (10000.0) & 1.65 & 33.4 & 33.4 & 16.6 & 16.6 \\
        & Random0-STL & 11.0 (25.0) & 5.0 (12.0) & 205.5 (4234.1) & 1.54 & 33.3 & 33.0 & 16.9 & 16.7 \\
        \bottomrule
    \end{tabular}
\end{table}

\noindent\textbf{MLTL.}  
The first benchmark suite in~\citet{LiVR22} contains three subsets of MLTL formulas:
\begin{itemize}[nosep]
  \item \emph{NASA-Boeing}: 63 satisfiable formulas extracted from real-world avionics systems;
  \item \emph{Random}: 4,000 randomly generated formulas with realistic patterns and random integer time intervals;
  \item \emph{Random0}: 4,000 randomly generated formulas whose time intervals all start at 0.
\end{itemize}
All MLTL formulas use only Boolean variables, stressing temporal operators over large horizons without linear real arithmetic.

\noindent \textbf{STL.} The second benchmark suite comprises three groups of STL formulas:
\begin{itemize}[nosep]
  \item nine requirements collected by \citet{MelaniBC25} plus the running example $\Phi_w$;
  \item \emph{Random-STL}: 2,000 randomly generated formulas of varying length;
  \item \emph{Random0-STL}: 2,000 random formulas of varying length with all intervals starting at 0.
\end{itemize}

\noindent Note that the satisfiability of all generated formulas is not controlled beforehand and is ultimately determined by the solving engines. The generation process for the random STL benchmark suites is detailed in
\appendixreference{app:stl-bench}{C}.

\subsection{Results}

\begin{figure}
    \centering
    \begin{subfigure}[b]{0.3\textwidth}
        \centering
        \caption{NASA-Boeing suite}
        \label{fig:mltl-survival-nasa-boeing}
        \includegraphics[width=\linewidth]{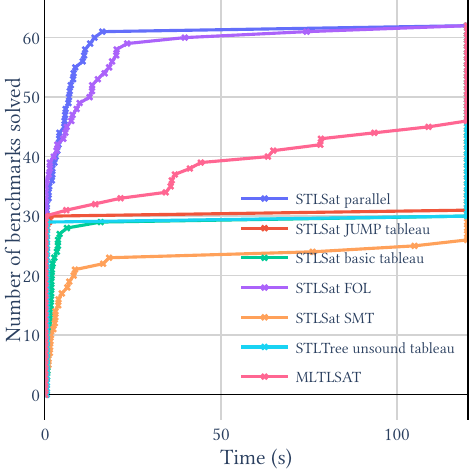}
    \end{subfigure}%
    \begin{subfigure}[b]{0.3\textwidth}
        \centering
        \caption{Random suite}
        \label{fig:mltl-survival-random}
        \includegraphics[width=\linewidth]{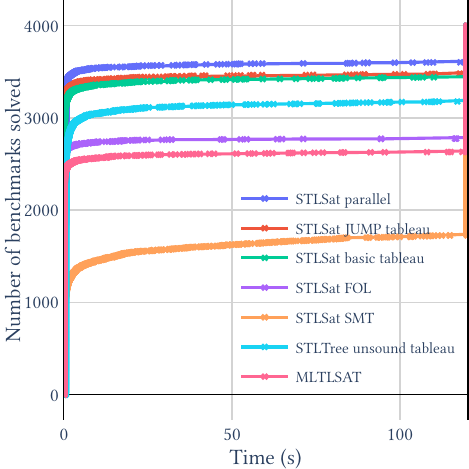}
    \end{subfigure}%
    \begin{subfigure}[b]{0.3\textwidth}
        \centering
        \caption{Random0 suite}
        \label{fig:mltl-survival-random0}
        \includegraphics[width=\linewidth]{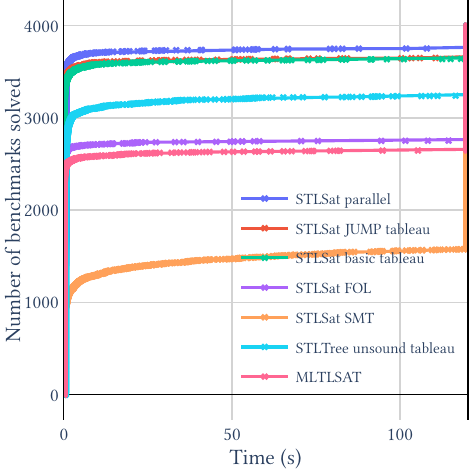}
    \end{subfigure}
    \caption{Survival plots for MLTL benchmarks \cite{LiVR22}. Higher is better.}
    \label{fig:mltl-plots}
\end{figure}

\begin{table}
    \centering
    \caption{Fastest \stlsat{} engine by benchmark and satisfiability class. Upper rows report counts and lower rows percentages among instances solved by at least one engine ($n$); ties are counted separately.}
    \label{tab:wins}
    \footnotesize
    \begin{tabular}{l l cccc r cccc r}
        & & \multicolumn{5}{c}{\textbf{sat}} & \multicolumn{5}{c}{\textbf{unsat}} \\
        \cmidrule(lr){3-7} \cmidrule(lr){8-12}
        & \textbf{benchmark} & \textbf{tableau} & \textbf{FOL} & \textbf{SMT} & \textbf{tie} & \textbf{$n$}
                             & \textbf{tableau} & \textbf{FOL} & \textbf{SMT} & \textbf{tie} & \textbf{$n$} \\
        \midrule
        \multirow{6}{*}{\rotatebox[origin=c]{90}{MLTL}}
        & \multirow{2}{*}{nasa-boeing} & 27 & \textbf{35} & 0 & 0 & 62 & -- & -- & -- & -- & 0 \\
        &                              & 43.5\% & \textbf{56.5\%} & 0.0\% & 0.0\% &  & -- & -- & -- & -- &  \\
        \cmidrule(lr){2-12}
        & \multirow{2}{*}{random} & \textbf{3166} & 226 & 14 & 17 & 3423 & \textbf{153} & 23 & 3 & 0 & 179 \\
        &                              & \textbf{92.5\%} & 6.6\% & 0.4\% & 0.5\% &  & \textbf{85.5\%} & 12.8\% & 1.7\% & 0.0\% &  \\
        \cmidrule(lr){2-12}
        & \multirow{2}{*}{random0} & \textbf{3278} & 190 & 3 & 16 & 3487 & \textbf{217} & 41 & 0 & 3 & 261 \\
        &                              & \textbf{94.0\%} & 5.4\% & 0.1\% & 0.5\% &  & \textbf{83.1\%} & 15.7\% & 0.0\% & 1.2\% &  \\
        \midrule
        \multirow{4}{*}{\rotatebox[origin=c]{90}{STL}}
        & \multirow{2}{*}{random} & \textbf{702} & 143 & 146 & 25 & 1016 & \textbf{227} & 121 & 22 & 43 & 413 \\
        &                              & \textbf{69.1\%} & 14.1\% & 14.4\% & 2.4\% &  & \textbf{55.0\%} & 29.3\% & 5.3\% & 10.4\% &  \\
        \cmidrule(lr){2-12}
        & \multirow{2}{*}{random0} & \textbf{682} & 36 & 81 & 19 & 818 & \textbf{284} & 108 & 40 & 22 & 454 \\
        &                              & \textbf{83.4\%} & 4.4\% & 9.9\% & 2.3\% &  & \textbf{62.6\%} & 23.8\% & 8.8\% & 4.8\% &  \\
        \bottomrule
    \end{tabular}
\end{table}

\noindent \textbf{MLTL.} We compare \stlsat{} with \stltree{} and the FOL encoding of \cite{LiVR22} (MLTLSAT). We show the results as survival plots in Fig.~\ref{fig:mltl-plots}, with a 120 s timeout.
In Table~\ref{tab:wins} we present the number of instances in which each \stlsat{} solving engine was strictly the fastest, grouped by benchmark and formula satisfiability for both MLTL and STL suites.
On the NASA-Boeing suite, the FOL encoding (\S~\ref{sec:fol-encoding}) dominates: it solved 62 benchmarks within 80 s, reporting a parsing error on one.
MLTLSAT solved 46 benchmarks within 108 s and timed out on 14 launching parsing errors on 3.
This advantage comes from a smaller generated encoding with fewer quantifiers: these formulas---derived from manually written requirement sets---contain only $\evenop{}$ and $\globop{}$ rather than $\untilop{}$ operators, and can benefit from our \emph{ad hoc} encoding of $\evenop{}$ and $\globop{}$.

On the other hand, tableau-based tools perform better on randomly generated formulas, with \stlsat{} outperforming STLTree despite using stricter jumping rules, thanks to the efficient algorithm implementation.
In fact, formulas in these suites contain many temporal operators with wide time intervals, which can greatly benefit from the tableau's \textsf{JUMP} rule.
The SMT encoding (\S~\ref{sec:smt-encoding}) performed worst on average, consistent with the findings of \cite{MelaniBC25}, although it was the only \stlsat{} engine able to solve the \emph{Aerospace} benchmark within the time limit.

Executing \stlsat{}'s three solving engines in parallel brings a slight benefit in random benchmarks, because there are some formulas for which one particular engine performs better than the others.
In the NASA-Boeing set, instead, the FOL encoding terminates first on the large majority of instances (although by a small margin), and parallel execution offsets the overhead of running all engines.

\begin{table}
    \centering
    \caption{Results on STL benchmarks from~\cite{MelaniBC25}. Columns 2 to 4 contain, respectively, number of requirements, maximum time Horizon, and temporal operators.
    Column R gives the result, shared by all approaches ($\top =$ ``satisfiable'', $\bot =$ ``unsatisfiable'').
    The remaining columns give the solution time (s) for each approach, where TO denotes a timeout at 120\,s. Bold marks the fastest solver per row.}
    \label{tab:stl-paper-bench}
    \footnotesize
    \begin{tabular}{l r r l c R{1.3cm} R{1.3cm} R{1.3cm} R{1.3cm} R{1.3cm}}
        \toprule
        \textbf{Benchmark} & \textbf{\#} & \textbf{H} & \textbf{Op.} & \textbf{R} & \textbf{STLTree tableau} & \textbf{\stlsat{} SMT} & \textbf{\stlsat{} FOL} & \textbf{\stlsat{} basic t.} & \textbf{\stlsat{} \textsf{JUMP} t.} \\
        \midrule
        Car \cite{BaeLee19} & 4 & 100 & $\globop$,$\evenop$,$\untilop$ & $\top$ & 0.5328 & 0.0238 & 0.0212 & 0.0114 & \textbf{0.0050} \\
        Thermostat \cite{BaeLee19} & 4 & 40 & $\globop$,$\evenop$,$\untilop$,$\releaseop$ & $\top$ & 0.4930 & 0.0186 & 0.0310 & 0.0060 & \textbf{0.0044} \\
        Watertank \cite{BaeLee19} & 4 & 50 & $\globop$,$\evenop$ & $\top$ & 0.5044 & 0.0174 & 0.0194 & \textbf{0.0116} & 0.0146 \\
        Railroad \cite{BaeLee19} & 4 & 100 & $\globop$,$\evenop$ & $\bot$ & 5.4144 & \textbf{0.0186} & 0.0582 & 0.0538 & 0.0258 \\
        Battery \cite{BaeLee19} & 4 & 64 & $\globop$,$\evenop$,$\untilop$ & $\top$ & 0.5940 & 0.0240 & 0.0852 & 0.0158 & \textbf{0.0088} \\
        PCV \cite{pcv} & 9 & 1,250 & $\globop$,$\evenop$ & $\top$ & 2.2950 & 1.0698 & \textbf{0.0190} & 22.1564 & 25.0328 \\
        MTL \cite{Benchmarks_Temporal_Logic_Requirements} & 10 & 2,020 & $\globop$,$\evenop$ & $\bot$ & 1.4636 & 0.1914 & \textbf{0.0178} & 0.2798 & 0.2834 \\
        CPS \cite{BoufaiedJBBP21} & 37 & 13,799 & $\globop$,$\evenop$,$\untilop$ & $\bot$ & 3.1226 & TO & \textbf{0.0198} & TO & TO \\
        Aerospace \cite{MelaniBC25} & 31 & 1,010 & $\globop$,$\evenop$ & $\top$ & \textbf{8.7816} & 37.5802 & TO & TO & TO \\
        Irrigation $\Phi_w$, \S\ref{sec:intro} & 3 & 160 & $\globop$,$\evenop$ & $\bot$ & 0.6774 & \textbf{0.0306} & TO & 0.0698 & 0.0404 \\
        \bottomrule
    \end{tabular}
\end{table}

\noindent \textbf{STL.} Table~\ref{tab:stl-paper-bench} shows our comparison on benchmarks from \cite{MelaniBC25}, plus the running example $\Phi_w$.
\stlsat{} outperforms STLTree in some benchmarks, but times out in the largest ones;
on the other hand, the FOL encoding quickly solves CPS, and the SMT encoding solves Aerospace without timing out.

\begin{figure}
    \centering
    \begin{subfigure}[b]{0.3\textwidth}
        \centering
        \caption{Random suite}
        \label{fig:stl-survival-random}
        \includegraphics[width=\linewidth]{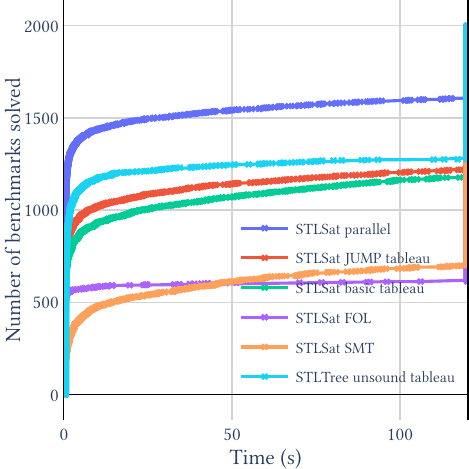}
    \end{subfigure}%
    \begin{subfigure}[b]{0.3\textwidth}
        \centering
        \caption{Random0 suite}
        \label{fig:stl-survival-random0}
        \includegraphics[width=\linewidth]{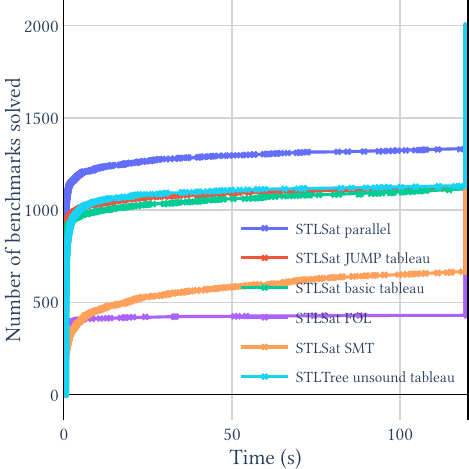}
    \end{subfigure}
    \caption{Survival plots for random STL benchmarks. Higher is better.}
    \label{fig:stl-plots}
\end{figure}

Fig.~\ref{fig:stl-plots} shows survival plots for randomly generated STL formulas.
\stlsat{}'s tableau performs slightly worse than \stltree{}, but much better than the FOL and SMT encodings. 
Table~\ref{tab:wins} suggests that the tableau is the fastest engine on most instances, while FOL and SMT are each fastest on smaller subsets of formulas. In particular, FOL wins on most of the NASA-Boeing benchmarks, and wins more frequently in unsatisfiable instances than satisfiable ones.
Consequently, parallel execution of the \stlsat{} engines yields further computational gains.

\paragraph{\textbf{Ablation Study}}

Our \textsf{JUMP} rule is a sound and complete alternative to the \oldjump{} rule~\cite{MelaniBC25}, although it is generally applicable in fewer cases.
To assess its practical impact, we compare the $\stlsat{}$ tableau  against the basic tableau with \textsf{JUMP} disabled (where only \textsf{STEP} rules execute) across all benchmarks from \S\ref{sec:benchmarks}.

Fig.~\ref{fig:plots-ablation} compares their solving times: most points lie below the bisector, showing that \textsf{JUMP} generally improves performance. 
Furthermore, instances where jumping provides no benefit from the optimization remain close to the bisector, indicating that computing and checking the jump conditions introduces negligible overhead.

\begin{figure}
    \centering
    \begin{subfigure}[b]{0.3\textwidth}
        \centering
        \caption{NASA-Boeing (MLTL)}
        \label{fig:mltl-ablation-nasa-boeing}
        \includegraphics[width=\linewidth]{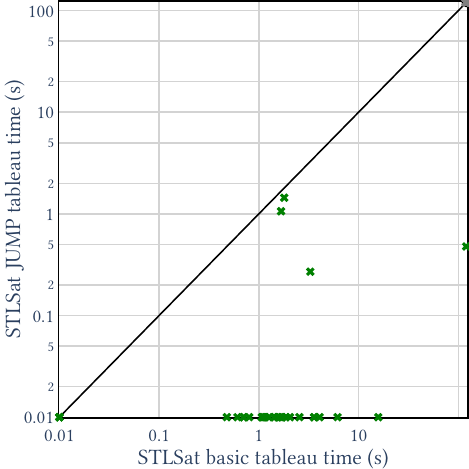}
    \end{subfigure}%
    \begin{subfigure}[b]{0.3\textwidth}
        \centering
        \caption{Random (MLTL)}
        \label{fig:mltl-ablation-random}
        \includegraphics[width=\linewidth]{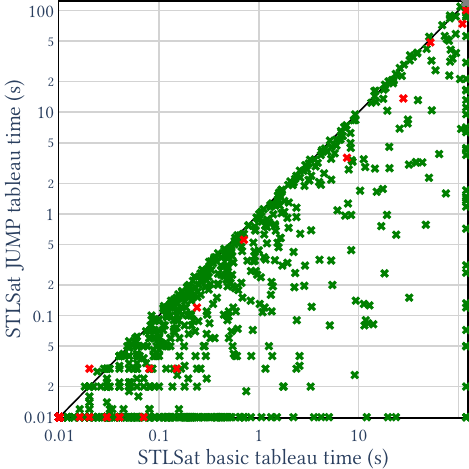}
    \end{subfigure}%
    \begin{subfigure}[b]{0.3\textwidth}
        \centering
        \caption{Random0 (MLTL)}
        \label{fig:mltl-ablation-random0}
        \includegraphics[width=\linewidth]{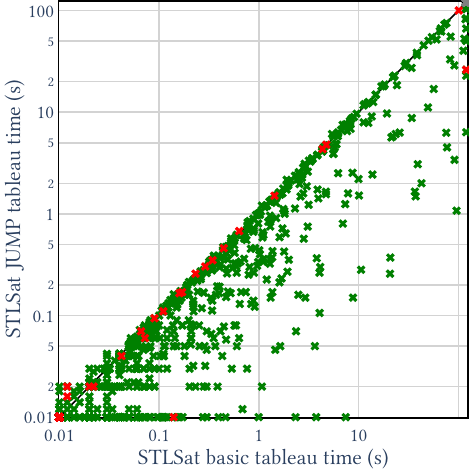}
    \end{subfigure}

    \begin{subfigure}[b]{0.3\textwidth}
        \centering
        \caption{Random (STL)}
        \label{fig:stl-ablation-random}
        \includegraphics[width=\linewidth]{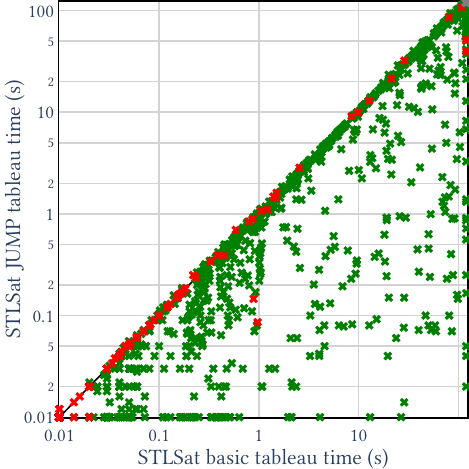}
    \end{subfigure}%
    \begin{subfigure}[b]{0.3\textwidth}
        \centering
        \caption{Random0 (STL)}
        \label{fig:stl-ablation-random0}
        \includegraphics[width=\linewidth]{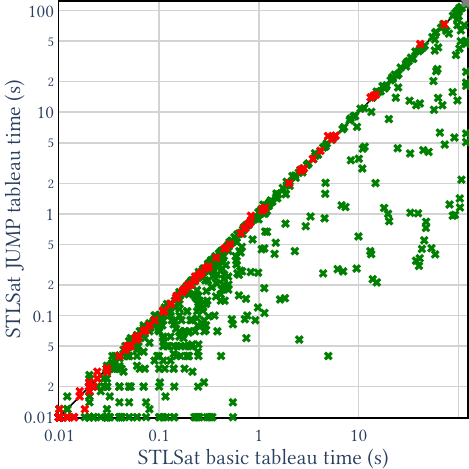}
    \end{subfigure}

    \caption{Scatter plots for the \stlsat{} ablation study on MLTL and STL benchmarks.}
    \label{fig:plots-ablation}
\end{figure}

\section{Conclusions}
\label{sec:conclusions}

We presented \stlsat{}, an open-source Rust tool for satisfiability checking of bounded discrete-time STL.
We identified an unsound optimization in the tableau of~\citet{MelaniBC25} and replaced it with a sound and complete rule.
\stlsat{} combines this tableau with two alternative encodings: a FOL translation and a quantifier-free SMT encoding, which can be executed in parallel as a portfolio solver.

Our extensive benchmark suite (covering both STL and MLTL, with real-world and randomly generated instances) shows that \stlsat{} performs comparably to \stltree{}, and often outperforms state-of-the-art SMT encodings such as MLTLSAT.  
Furthermore, running all \stlsat{} solving engines in parallel yields the best overall performance.
Moreover, we argued how unsat core extraction facilitates specification mining, redundancy elimination, and traceability in CPS development.

Future work will focus on temporal analysis for \emph{unsat cores}, support for unbounded STL operators and continuous-time semantics, tighter integration with model-checking, falsification, and specification mining, and the development of visual debugging interfaces that display counterexample signals and core constraints.

\section*{Acknowledgements}
Partially funded by the European Union.
Views and opinions expressed are those of the authors only and do not necessarily reflect those of the European Union or the European Health and Digital Executive Agency (HADEA).
Neither the European Union nor the granting authority can be held responsible for them. RobustifAI project, ID 101212818; VASSAL project, ID 101160022; MoSaiC project, ID 101103010.
Partially funded by WWTF project ICT22-023 (TAIGER).

\bibliographystyle{plainnat}
\bibliography{stlsat}

@inproceedings{MalerN04,
  author       = {Oded Maler and
                  Dejan Nickovic},
  _editor       = {Yassine Lakhnech and
                  Sergio Yovine},
  title        = {Monitoring Temporal Properties of Continuous Signals},
  _booktitle    = {Formal Techniques, Modelling and Analysis of Timed and Fault-Tolerant
                  Systems, Joint International Conferences on Formal Modelling and Analysis
                  of Timed Systems, {FORMATS} 2004 and Formal Techniques in Real-Time
                  and Fault-Tolerant Systems, {FTRTFT} 2004, Grenoble, France, September
                  22-24, 2004, Proceedings},
  booktitle    = {{FORMATS/FTRTFT}'04},
  series       = {LNCS},
  volume       = {3253},
  pages        = {152--166},
  publisher    = {Springer},
  year         = {2004},
  _url          = {https://doi.org/10.1007/978-3-540-30206-3\_12},
  doi          = {10.1007/978-3-540-30206-3\_12},
}

@inproceedings{BartocciMNY23,
  author       = {Ezio Bartocci and
                  Leonardo Mariani and
                  Dejan Nickovic and
                  Drishti Yadav},
  title        = {Property-Based Mutation Testing},
  _booktitle    = {Proc. of {ICST} 2023: the {IEEE} Conference on Software Testing, Verification and Validation},
  booktitle    = {{ICST} '23},
  pages        = {222--233},
  publisher    = {{IEEE}},
  year         = {2023},
  _url          = {https://doi.org/10.1109/ICST57152.2023.00029},
  doi          = {10.1109/ICST57152.2023.00029}
}

@inproceedings{BartocciFMN18,
  author       = {Ezio Bartocci and
                  Thomas Ferr{\`{e}}re and
                  Niveditha Manjunath and
                  Dejan Nickovic},
  _editor       = {Maria Prandini and
                  Jyotirmoy V. Deshmukh},
  title        = {Localizing Faults in {Simulink/Stateflow} Models with {STL}},
  _booktitle    = {Proc. of {HSCC} 2018: the 21st International Conference on Hybrid Systems: Computation and Control (part of {CPS} Week)},
  booktitle    = {{HSCC} '18},
  pages        = {197--206},
  publisher    = {{ACM}},
  year         = {2018},
  _url          = {https://doi.org/10.1145/3178126.3178131},
  doi          = {10.1145/3178126.3178131}
}

@inproceedings{SankaranarayananF12,
  author       = {Sriram Sankaranarayanan and
                  Georgios Fainekos},
  _editor       = {Thao Dang and
                  Ian M. Mitchell},
  title        = {Falsification of temporal properties of hybrid systems using the cross-entropy
                  method},
  _booktitle    = {Hybrid Systems: Computation and Control (part of {CPS} Week 2012),
                  HSCC'12, Beijing, China, April 17-19, 2012},
  booktitle    = {{HSCC}'12},
  pages        = {125--134},
  publisher    = {{ACM}},
  year         = {2012},
  doi          = {10.1145/2185632.2185653}
}

@article{BartocciBNS15,
  author       = {Ezio Bartocci and
                  Luca Bortolussi and
                  Laura Nenzi and
                  Guido Sanguinetti},
  title        = {System design of stochastic models using robustness of temporal properties},
  journal      = {Theor. Comput. Sci.},
  volume       = {587},
  pages        = {3--25},
  year         = {2015},
  doi          = {10.1016/J.TCS.2015.02.046}
}

@inproceedings{JaksicBGKNN15,
  author       = {Stefan Jaksic and
                  Ezio Bartocci and
                  Radu Grosu and
                  Reinhard Kloibhofer and
                  Thang Nguyen and
                  Dejan Nickovic},
  title        = {From signal temporal logic to {FPGA} monitors},
  booktitle    = {{MEMOCODE} '15},
  pages        = {218--227},
  publisher    = {{IEEE}},
  year         = {2015},
  doi          = {10.1109/MEMCOD.2015.7340489}
}

@inproceedings{MLTL,
    author = {Jianwen, Li and Moshe Y. Vardi and Kristin Y. Rozier},
    title = {Satisfiability Checking for Mission-Time LTL},
    editor = {Dillig, I. and Tasiran, S.},
    booktitle = {Computer Aided Verification. CAV 2019.},
    series   = {LNCS},
    volume = {11562},
    publisher    = {Springer},
    year = {2019},
    doi = {https://doi.org/10.1007/978-3-030-25543-5_1}    
}

@article{BaeLee19,
    author = {Kyungmin Bae and Jia Lee},
    title = {Bounded model checking of signal temporal logic properties using syntactic separation},
    journal = {Proc. ACM Program. Lang.},
    volume = {3},
    pages = {1-30},
    publisher = {ACM},
    series   = {POPL},
    year = {2019},
    doi = {10.1145/3290364}    
}

@inproceedings{LeeYB21,
  author       = {Jia Lee and
                  Geunyeol Yu and
                  Kyungmin Bae},
  title        = {Efficient {SMT}-Based Model Checking for Signal Temporal Logic},
  booktitle    = {{ASE}'21},
  pages        = {343--354},
  publisher    = {{IEEE}},
  year         = {2021},
  _url          = {https://doi.org/10.1109/ASE51524.2021.9678719},
  doi          = {10.1109/ASE51524.2021.9678719},
}

@inproceedings{Reynolds16,
    author = {Mark Reynolds},
    title = {A New Rule for {LTL} Tableaux},
    booktitle = {GandALF'16},
    series   = {EPCTS},
    volume = {26},
    pages = {287-301},
    year = {2016},
    doi = {10.4204/EPTCS.226.20}
}

@article{GeattiGMR21,
  author       = {Luca Geatti and
                  Nicola Gigante and
                  Angelo Montanari and
                  Mark Reynolds},
  title        = {One-pass and tree-shaped tableau systems for {TPTL} and {TPTL}\({}_{\mbox{b}}\)+Past},
  journal      = {Inf. Comput.},
  volume       = {278},
  pages        = {104599},
  year         = {2021},
  _url          = {https://doi.org/10.1016/j.ic.2020.104599},
  doi          = {10.1016/J.IC.2020.104599},
}

@article{GeattiGMV24,
  author       = {Luca Geatti and
                  Nicola Gigante and
                  Angelo Montanari and
                  Gabriele Venturato},
  title        = {{SAT} Meets Tableaux for Linear Temporal Logic Satisfiability},
  journal      = {J. Autom. Reason.},
  volume       = {68},
  number       = {2},
  pages        = {6},
  year         = {2024},
  doi          = {10.1007/S10817-023-09691-1},
}

@inproceedings{vacuity,
  author={Dokhanchi, Adel and Yaghoubi, Shakiba and Hoxha, Bardh and Fainekos, Georgios},
  booktitle={2017 13th IEEE Conference on Automation Science and Engineering (CASE)}, 
  title={Vacuity aware falsification for MTL request-response specifications}, 
  year={2017},
  volume={},
  number={},
  pages={1332-1337},
  doi={10.1109/COASE.2017.8256286}
}

@inproceedings{RoehmHM17,
  author       = {Hendrik Roehm and
                  Thomas Heinz and
                  Eva Charlotte Mayer},
  _editor       = {Rupak Majumdar and
                  Viktor Kuncak},
  title        = {STLInspector: {STL} Validation with Guarantees},
  booktitle    = {{CAV}'17},
  series       = {LNCS},
  volume       = {10426},
  pages        = {225--232},
  publisher    = {Springer},
  year         = {2017},
  doi          = {10.1007/978-3-319-63387-9\_11}
}

@inproceedings{BersaniRP13,
  author       = {Marcello M. Bersani and
                  Matteo Rossi and
                  Pierluigi {San Pietro}},
  _editor       = {C{\'{e}}sar S{\'{a}}nchez and
                  Kristen Brent Venable and
                  Esteban Zim{\'{a}}nyi},
  title        = {A Tool for Deciding the Satisfiability of Continuous-Time Metric Temporal
                  Logic},
  booktitle    = {TIME'13},
  pages        = {99--106},
  publisher    = {{IEEE} Computer Society},
  year         = {2013},
  _url          = {https://doi.org/10.1109/TIME.2013.20},
  doi          = {10.1109/TIME.2013.20},
}

@article{BartocciMMMN21,
  author       = {Ezio Bartocci and
                  Niveditha Manjunath and
                  Leonardo Mariani and
                  Cristinel Mateis and
                  Dejan Nickovic},
  title        = {CPSDebug: Automatic failure explanation in {CPS} models},
  journal      = {Int. J. Softw. Tools Technol. Transf.},
  volume       = {23},
  number       = {5},
  pages        = {783--796},
  year         = {2021},
  doi          = {10.1007/S10009-020-00599-4}
}

@article{BoufaiedJBBP21,
  author       = {Chaima Boufaied and
                  Maris Jukss and
                  Domenico Bianculli and
                  Lionel C. Briand and
                  Yago Isasi Parache},
  title        = {Signal-Based Properties of Cyber-Physical Systems: Taxonomy and Logic-based
                  Characterization},
  journal      = {J. Syst. Softw.},
  volume       = {174},
  pages        = {110881},
  year         = {2021},
  doi          = {10.1016/J.JSS.2020.110881}
}

@incollection{BartocciDDFMNS18,
  author       = {Ezio Bartocci and
                  Jyotirmoy V. Deshmukh and
                  Alexandre Donz{\'{e}} and
                  Georgios Fainekos and
                  Oded Maler and
                  Dejan Nickovic and
                  Sriram Sankaranarayanan},
  _editor       = {Ezio Bartocci and
                  Yli{\`{e}}s Falcone},
  title        = {Specification-Based Monitoring of Cyber-Physical Systems: {A} Survey
                  on Theory, Tools and Applications},
  booktitle    = {Lectures on Runtime Verification - Introductory and Advanced Topics},
  series       = {LNCS},
  volume       = {10457},
  pages        = {135--175},
  publisher    = {Springer},
  year         = {2018},
  _url          = {https://doi.org/10.1007/978-3-319-75632-5\_5},
  doi          = {10.1007/978-3-319-75632-5\_5}
}

@article{BartocciMNN22,
  author       = {Ezio Bartocci and
                  Cristinel Mateis and
                  Eleonora Nesterini and
                  Dejan Nickovic},
  title        = {Survey on mining signal temporal logic specifications},
  journal      = {Inf. Comput.},
  volume       = {289},
  number       = {Part},
  pages        = {104957},
  year         = {2022},
  _url          = {https://doi.org/10.1016/j.ic.2022.104957},
  doi          = {10.1016/J.IC.2022.104957}
}

@inproceedings{Z3,
    author = {de Moura, Leonardo and Bjørner, Nicolaj},
    title = {Z3: An Efficient {SMT} Solver},
    _editor = {Ramakrishnan, C.R. and Rehof, J.},
    booktitle = {TACAS'08.},
    series   = {LNCS},
    volume = {4963},
    publisher    = {Springer},
    year = {2008},
    doi = {https://doi.org/10.1007/978-3-540-78800-3_24}    
}

@inproceedings{Benchmarks_Temporal_Logic_Requirements,
  author    = {Bardh Hoxha and Houssam Abbas and Georgios Fainekos},
  title     = {Benchmarks for Temporal Logic Requirements for Automotive Systems},
  booktitle = {ARCH14-15},
  _editor    = {Goran Frehse and Matthias Althoff},
  series    = {EPiC Series in Computing},
  volume    = {34},
  publisher = {EasyChair},
  doi       = {10.29007/xwrs},
  pages     = {25-30},
  year      = {2015}}

@inproceedings{pcv,
author = {Jin, Xiaoqing and Deshmukh, Jyotirmoy V. and Kapinski, James and Ueda, Koichi and Butts, Ken},
title = {Powertrain control verification benchmark},
year = {2014},
isbn = {9781450327329},
publisher = {ACM},
address = {New York, USA},
doi = {10.1145/2562059.2562140},
booktitle = {{HSCC}'14},
pages = {253–262},
numpages = {10},
}

@article{LiVR22,
  author       = {Jianwen Li and
                  Moshe Y. Vardi and
                  Kristin Y. Rozier},
  title        = {Satisfiability checking for Mission-time {LTL} {(MLTL)}},
  journal      = {Inf. Comput.},
  volume       = {289},
  number       = {Part},
  pages        = {104923},
  year         = {2022},
  doi          = {10.1016/J.IC.2022.104923}
}

@article{AlurFH96,
  author       = {Rajeev Alur and
                  Tom{\'{a}}s Feder and
                  Thomas A. Henzinger},
  title        = {The Benefits of Relaxing Punctuality},
  journal      = {J. {ACM}},
  volume       = {43},
  number       = {1},
  pages        = {116--146},
  year         = {1996},
  doi          = {10.1145/227595.227602}
}

@inproceedings{RamanDMMSS14,
  author       = {Vasumathi Raman and
                  Alexandre Donz{\'{e}} and
                  Mehdi Maasoumy and
                  Richard M. Murray and
                  Alberto L. Sangiovanni{-}Vincentelli and
                  Sanjit A. Seshia},
  title        = {Model predictive control with signal temporal logic specifications},
  booktitle    = {{CDC}'14},
  pages        = {81--87},
  publisher    = {{IEEE}},
  year         = {2014},
  doi          = {10.1109/CDC.2014.7039363}
}

@article{BersaniRP15,
  author       = {Marcello M. Bersani and
                  Matteo Rossi and
                  Pierluigi {San Pietro}},
  title        = {An {SMT}-based approach to satisfiability checking of {MITL}},
  journal      = {Inf. Comput.},
  volume       = {245},
  pages        = {72--97},
  year         = {2015},
  doi          = {10.1016/J.IC.2015.06.007}
}

@article{BersaniRP16,
  author       = {Marcello M. Bersani and
                  Matteo Rossi and
                  Pierluigi {San Pietro}},
  title        = {A tool for deciding the satisfiability of continuous-time metric temporal
                  logic},
  journal      = {Acta Informatica},
  volume       = {53},
  number       = {2},
  pages        = {171--206},
  year         = {2016},
  doi          = {10.1007/S00236-015-0229-Y}
}

@article{HirshfeldR05,
  author       = {Yoram Hirshfeld and
                  Alexander Moshe Rabinovich},
  title        = {Timer formulas and decidable metric temporal logic},
  journal      = {Inf. Comput.},
  volume       = {198},
  number       = {2},
  pages        = {148--178},
  year         = {2005},
  doi          = {10.1016/J.IC.2004.12.002}
}

@inproceedings{ReinbacherRS14,
  author       = {Thomas Reinbacher and
                  Kristin Y. Rozier and
                  Johann Schumann},
  _editor       = {Erika {\'{A}}brah{\'{a}}m and
                  Klaus Havelund},
  title        = {Temporal-Logic Based Runtime Observer Pairs for System Health Management of Real-Time Systems},
  booktitle    = {{TACAS}'14},
  series       = {LNCS},
  volume       = {8413},
  pages        = {357--372},
  publisher    = {Springer},
  year         = {2014},
  doi          = {10.1007/978-3-642-54862-8\_24}
}

@misc{mltlsat,
  author = {Jianwen Li},
  title = {mltlsat},
  howpublished = {GitHub.com},
  url = {https://github.com/lijwen2748/mltlsat},
  year = {2019}
}

@article{Dokhanchi,
author = {Dokhanchi, Adel and Hoxha, Bardh and Fainekos, Georgios},
title = {Formal Requirement Debugging for Testing and Verification of Cyber-Physical Systems},
year = {2017},
issue_date = {March 2018},
publisher = {Association for Computing Machinery},
address = {New York, NY, USA},
volume = {17},
number = {2},
issn = {1539-9087},
_url = {https://doi.org/10.1145/3147451},
doi = {10.1145/3147451},
journal = {ACM Trans. Embed. Comput. Syst.},
articleno = {34},
numpages = {26}
}

@misc{stltree,
  author = {Beatrice Melani and Ezio Bartocci and Michele Chiari},
  title = {STLTree},
  url = {https://github.com/beamelani/Consistency_Check},
  year = 2025
}

@article{MelaniBC25,
  title={A Tree-Shaped Tableau for Checking the Satisfiability of Signal Temporal Logic with Bounded Temporal Operators},
  author={Melani, Beatrice and Bartocci, Ezio and Chiari, Michele},
  journal={{ACM} Trans. Embed. Comput. Syst.},
  volume={24},
  number={5s},
  pages={1--26},
  year={2025},
  publisher={ACM New York, NY},
  doi={10.1145/3759917}
}

@inproceedings{CimattiGS07,
  author       = {Alessandro Cimatti and
                  Alberto Griggio and
                  Roberto Sebastiani},
  _editor       = {Jo{\~{a}}o Marques{-}Silva and
                  Karem A. Sakallah},
  title        = {A Simple and Flexible Way of Computing Small Unsatisfiable Cores in
                  {SAT} Modulo Theories},
  _booktitle    = {Theory and Applications of Satisfiability Testing - {SAT} 2007, 10th
                  International Conference, Lisbon, Portugal, May 28-31, 2007, Proceedings},
  booktitle    = {{SAT}'07},
  series       = {LNCS},
  volume       = {4501},
  pages        = {334--339},
  publisher    = {Springer},
  year         = {2007},
  _url          = {https://doi.org/10.1007/978-3-540-72788-0\_32},
  doi          = {10.1007/978-3-540-72788-0\_32},
}

@article{Schuppan16,
  author       = {Viktor Schuppan},
  title        = {Extracting unsatisfiable cores for {LTL} via temporal resolution},
  journal      = {Acta Informatica},
  volume       = {53},
  number       = {3},
  pages        = {247--299},
  year         = {2016},
  _url          = {https://doi.org/10.1007/s00236-015-0242-1},
  doi          = {10.1007/S00236-015-0242-1},
}

@book{cormen2022introduction,
  title={Introduction to algorithms},
  author={Cormen, Thomas H and Leiserson, Charles E and Rivest, Ronald L and Stein, Clifford},
  year={2022},
  publisher={MIT press}
}

@inproceedings{armando2004sat,
  author       = {Alessandro Armando and
                  Claudio Castellini and
                  Enrico Giunchiglia and
                  Marco Maratea},
  title        = {A {SAT}-based Decision Procedure for the {Boolean} Combination of Difference
                  Constraints},
  booktitle    = {{SAT} 2004},
  year         = {2004},
  url          = {http://www.satisfiability.org/SAT04/programme/71.pdf},
}

@inproceedings{DutertreM06,
  author       = {Bruno Dutertre and
                  Leonardo M. de Moura},
  _editor       = {Thomas Ball and
                  Robert B. Jones},
  title        = {A Fast Linear-Arithmetic Solver for {DPLL(T)}},
  booktitle    = {{CAV} 2006},
  series       = {LNCS},
  volume       = {4144},
  pages        = {81--94},
  publisher    = {Springer},
  year         = {2006},
  _url          = {https://doi.org/10.1007/11817963\_11},
  doi          = {10.1007/11817963\_11},
}

@article{ClarkeBRZ01,
  author       = {Edmund M. Clarke and
                  Armin Biere and
                  Richard Raimi and
                  Yunshan Zhu},
  title        = {Bounded Model Checking Using Satisfiability Solving},
  journal      = {Formal Methods Syst. Des.},
  volume       = {19},
  number       = {1},
  pages        = {7--34},
  year         = {2001},
  _url          = {https://doi.org/10.1023/A:1011276507260},
  doi          = {10.1023/A:1011276507260},
  timestamp    = {Fri, 13 Mar 2020 10:55:19 +0100},
  biburl       = {https://dblp.org/rec/journals/fmsd/ClarkeBRZ01.bib},
  bibsource    = {dblp computer science bibliography, https://dblp.org}
}

@inproceedings{deMouraU21,
  title = {The {Lean} 4 Theorem Prover and Programming Language},
  author = {de Moura, Leonardo and Ullrich, Sebastian},
  year = {2021},
  _isbn = {978-3-030-79875-8},
  publisher = {Springer-Verlag},
  address = {Berlin, Heidelberg},
  series = {LNCS},
  _url = {https://doi.org/10.1007/978-3-030-79876-5_37},
  doi = {10.1007/978-3-030-79876-5_37},
  booktitle = {CADE'21},
  _booktitle = {Automated Deduction – CADE 28: 28th International Conference on Automated Deduction, Virtual Event, July 12–15, 2021, Proceedings},
  pages = {625–635},
  numpages = {11}
}

@misc{openai2026gpt56,
  author       = {{OpenAI}},
  title        = {{GPT-5.6}: Frontier intelligence that scales with your ambition},
  year         = {2026},
  howpublished = {\url{https://openai.com/index/gpt-5-6/}},
  note         = {Accessed: 2026-08-28}
}

@misc{zenodo,
  author       = {Zamponi, Marco and
                  Chiari, Michele and
                  Bartocci, Ezio},
  title        = {Replication package for the paper ``{STLSat}---An
                   Improved Tableau for Satisfiability Checking of
                   {Signal Temporal Logic} Formulas''
                  },
  year         = 2026,
  publisher    = {Zenodo},
  doi          = {10.5281/zenodo.22746464},
  url          = {https://doi.org/10.5281/zenodo.22746464},
}

@misc{stlsat,
  author = {Zamponi, Marco and
            Chiari, Michele and
            Lammel, Florian},
  title = {{STLSat}},
  howpublished = {GitHub.com},
  url = {https://github.com/ZamponiMarco/stlsat},
  year = {2026}
}

@article{ZamponiLBC26,
  author       = {Marco Zamponi and
                  Florian Lammel and
                  Ezio Bartocci and
                  Michele Chiari},
  title        = {{STLSat}---{An} Improved Tableau for Satisfiability Checking of {Signal Temporal Logic} Formulas},
  journal      = {CoRR},
  volume       = {abs/2607.21081},
  year         = {2026},
  _url          = {https://doi.org/10.48550/arXiv.2607.21081},
  doi          = {10.48550/ARXIV.2607.21081},
  eprinttype   = {arXiv},
  eprint       = {2607.21081},
}

\ifthenelse{\boolean{appendix}}{%
  \appendix
  \section{Flaws in the \texorpdfstring{\oldjump{}}{JUMPO} rule}
\label{app:jump-flaws}

In this section, we show that the \oldjump{} rule proposed by \citet{MelaniBC25} undermines both soundness and completeness of the tableau.
We report the definition of the \oldjump{} rule almost verbatim from \cite{MelaniBC25} below.

\citet{MelaniBC25} manage formula superscripts in a slightly different way:
nested operators extracted by expansion rules are annotated with the interval of the parent formula, rather than the formula itself:
\begin{align*}
&\exp^t_I(\anybinop{[a,b]}{\varphi_1}{\varphi_2}) = \anybinop[I]{[a+t,b+t]}{\varphi_1}{\varphi_2} && \text{with } \mathsf{B} \in \set{\suntilop, \sreleaseop}
\end{align*}
where $I$ is the interval $[a,b]$ in Table~\ref{tab:expansion-rules},
and $\exp^t_I$ replaces $\exp^t_\psi$ in that table.
Note that, since such superscripts are not used in the basic tableau, this change does not interfere with its correctness.

The \oldjump{} rule can be applied to a poised node $u$ if 
all marked operators
$\msuntil[J]{[a,b]}{\varphi_1}{\varphi_2}$ or $\msrelease[J]{[a,b]}{\varphi_2}{\varphi_1}$ in $\Gamma(u)$,
with $J = [a_J, b_J]$ and $t(u) \geq b_J$, are such that
\begin{enumerate}[label=\roman*.]
\item $t(u) < b$, and
\item \label{item:oldjump-pcl-rule}
    for all temporal operators $\psi$ in $\varphi_1$ that are not nested in any other temporal operator we have $t(u) \geq a + I_u(\psi)$,
\end{enumerate}
where $I_u(\psi)$ is the upper endpoint of the interval of \(\psi\)
(e.g., $I_u(\glob{[a,b]}\psi') = b$).
Otherwise, only the \textsf{STEP} rule can be applied.

If the conditions above are met, then $u$ has one child $u'$ such that $t(u') = \min \set{t \in K(u) \mid t > t(u)}$,
where
\[
K(u) = \set{a, b \mid \anybinop[J]{[a,b]}{\varphi_1}{\varphi_2} \in \Gamma(u) \land \mathsf{B} \in \set{\suntilop, \sreleaseop, \marked{\suntilop}, \marked{\sreleaseop}} \land t(u) \geq b_J},
\]
and, setting $k = t(u') - t(u)$, we have
\begin{align*}
\Gamma(u') =
&\ \set{\anybinop[J]{I}{\varphi_1}{\varphi_2} \in \Gamma(u) \mid \mathsf{B} \in \set{\suntilop, \sreleaseop} \land t(u) \geq b_J} \\ 
&\cup
\set{\anybinop[J]{[a+k, b+k]}{\varphi_1}{\varphi_2} \mid \anybinop[J]{[a,b]}{\varphi_1}{\varphi_2} \in \Gamma(u) \land \mathsf{B} \in \set{\suntilop, \sreleaseop} \land t(u) < b_J} \\ 
&\cup
\set{\anybinop[J]{[a,b]}{\varphi_1}{\varphi_2} \mid \stlbinop[J]{\marked{\mathsf{B}}}{[a,b]}{\varphi_1}{\varphi_2} \in \Gamma(u) \land \mathsf{B} \in \set{\suntilop, \sreleaseop} \land t(u) \geq b_J \land t(u) < b} \\ 
&\cup
\set{\anybinop[J]{[a+k,b+k]}{\varphi_1}{\varphi_2} \mid \stlbinop[J]{\marked{\mathsf{B}}}{[a,b]}{\varphi_1}{\varphi_2} \in \Gamma(u) \land \mathsf{B} \in \set{\suntilop, \sreleaseop} \land t(u) < b_J \land t(u) < b} 
\end{align*}

The \oldjump{} rule has two main issues:
\begin{enumerate}
\item condition \ref{item:oldjump-pcl-rule} only takes into account one level of nesting of temporal operators, and only in the right-hand-side operand; and
\item the set $\Gamma(u')$ does not faithfully re-create the label that node $u'$ would have if it had been reached by repeatedly applying the \textsf{STEP} rule until reaching a node with $t = t(u) + k$, dropping formulas that can reveal inconsistencies.
\end{enumerate}

As shown in Section~\ref{sec:unsound-jump}, formula $\varphi_S$ witnesses the unsoundness of the \oldjump{}  rule mainly because of the second issue above.
The following formula illustrates the first issue:
\[
\varphi'_S =
\suntil{[0,4]}{(\glob{[0,0]} \glob{[5,5]} a)}{(\glob{[5,5]} \neg c)} \land \glob{[5,8]} c \land \glob{[8,8]} \neg a
\]

Fig.~\ref{fig:unsound-tree} shows part of the tableau rooted in $\varphi'_S$ that uses the \oldjump{} rule.
First, the tableau tries to satisfy the $\untilop$ operator at time 0 by expanding it to $\glob{[5,5]} \neg c$ in node $u_2$, but this leads to a contradiction with $\glob{[5,8]} c$.
It thus tries extracting $\psi = \glob{[0,0]} \glob{[5,5]} a$ in node $u_3$, and since here $I_u(\psi) = 0$, the conditions for the \oldjump{} rule apply.
Thus, the tableau jumps to $u_8$ which, as noted in Section~\ref{sec:unsound-jump}, does not contain $\glob{[8,8]} \neg c$, allowing the satisfaction of all operators.

The new \textsf{JUMP} rule that we present in Section~\ref{sec:jump} does not suffer from the same issue thanks to the \emph{soundness limit} (cf.\ Section~\ref{sec:jump-size-calculation}).
In fact, due to the $\untilop$ operator we have $(0,5) \in N(u)$ and due to $\glob{[8,8]} \neg a$ we have $(8,8) \in O(u)$.
Thus, we have $k^{\ast}_{\text{sound}}(u) \leq l_o - r_n = 8 - 5 = 3$:
the tableau may jump at most to a node $u'$ with $t(u') = 3$,
in which either we try to satisfy the $\untilop$ operator, encountering the contradiction between $\glob{[8,8]} \neg c$ and $\glob{[5,8]} c$,
or we postpone its satisfaction, revealing the contradiction between $\glob{[3,3]} \glob{[5,5]} a$ and $\glob{[8,8]} \neg a$.

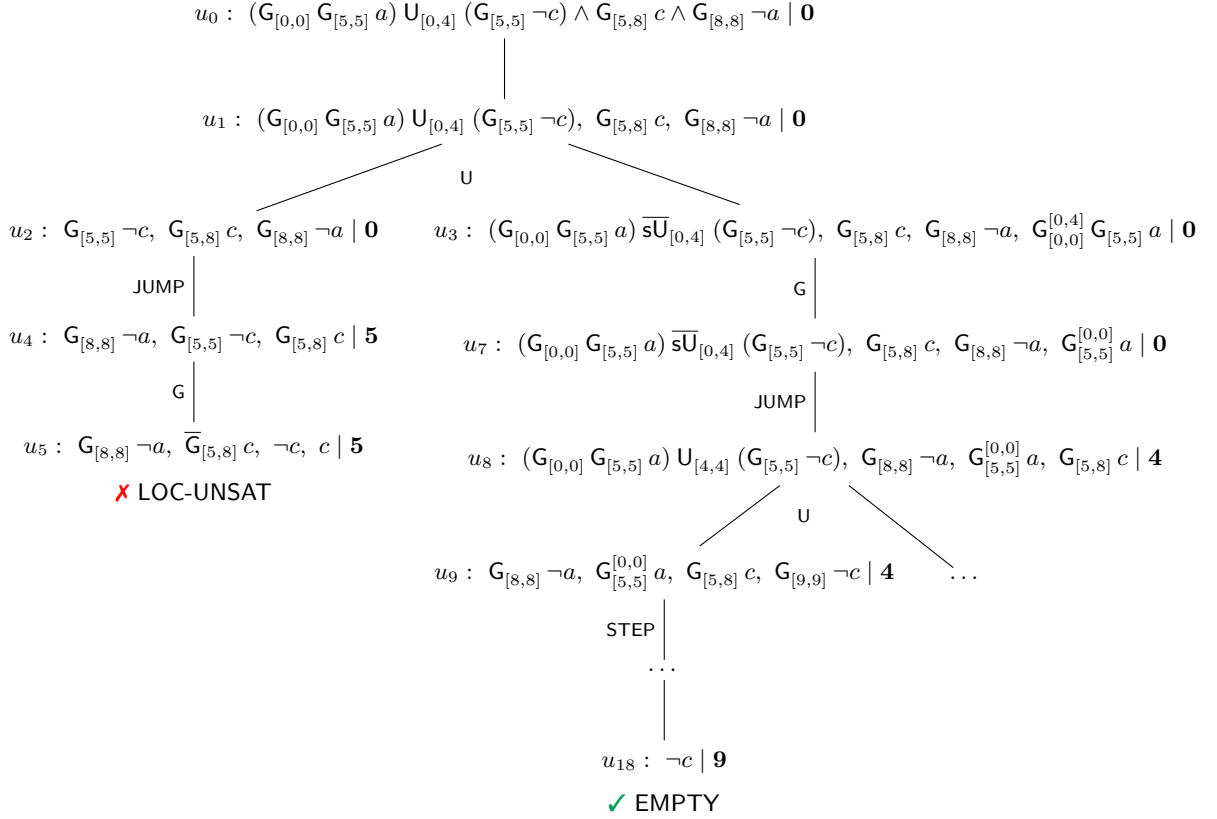
\begin{figure}
\centering
\small
\begin{forest}
  for tree={
    myleaf/.style={label=below:{\strut#1}},
    l sep=8mm,
    s sep=5mm
  },
  [{$u_0:\ \until{[0,4]}{(\glob{[0,0]}\glob{[5,5]}a)}{(\glob{[5,5]}\neg c)}\land \glob{[5,8]}c\land \glob{[8,8]}\neg a\mid \mathbf{0}$}
    [{$u_1:\ \until{[0,4]}{(\glob{[0,0]}\glob{[5,5]}a)}{(\glob{[5,5]}\neg c)},\ \glob{[5,8]}c,\ \glob{[8,8]}\neg a\mid \mathbf{0}$}
      [{$u_2:\ \glob{[5,5]}\neg c,\ \glob{[5,8]}c,\ \glob{[8,8]}\neg a\mid \mathbf{0}$},
        edge label={node[midway,left,xshift=50pt,font=\scriptsize\sffamily]{U}}
        [{$u_4:\ \glob{[8,8]}\neg a,\ \glob{[5,5]}\neg c,\ \glob{[5,8]}c\mid \mathbf{5}$},
          edge label={node[midway,left,font=\scriptsize\sffamily]{JUMP}}
          [{$u_5:\ \glob{[8,8]}\neg a,\ \mglob{[5,8]}c,\ \neg c,\ c\mid \mathbf{5}$},
          edge label={node[midway,left,font=\scriptsize\sffamily]{G}},
          myleaf={\xmark\ \textsf{LOC-UNSAT}}
          ]
        ]
      ]
      [{$u_3:\ \msuntil{[0,4]}{(\glob{[0,0]}\glob{[5,5]}a)}{(\glob{[5,5]}\neg c)},\ \glob{[5,8]}c,\ \glob{[8,8]}\neg a,\ \glob[{[0,4]}]{[0,0]}\glob{[5,5]}a\mid \mathbf{0}$},
        [{$u_7:\ \msuntil{[0,4]}{(\glob{[0,0]}\glob{[5,5]}a)}{(\glob{[5,5]}\neg c)},\ \glob{[5,8]}c,\ \glob{[8,8]}\neg a,\ \glob[{[0,0]}]{[5,5]}a\mid \mathbf{0}$},
          edge label={node[midway,left,font=\scriptsize\sffamily]{G}},
          [{$u_8:\ \until{[4,4]}{(\glob{[0,0]}\glob{[5,5]}a)}{(\glob{[5,5]}\neg c)},\ \glob{[8,8]}\neg a,\ \glob[{[0,0]}]{[5,5]}a,\ \glob{[5,8]}c \mid \mathbf{4}$},
            edge label={node[midway,left,font=\scriptsize\sffamily]{JUMP}}
            [{$u_9:\ \glob{[8,8]}\neg a,\ \glob[{[0,0]}]{[5,5]}a,\ \glob{[5,8]}c,\ \glob{[9,9]}\neg c\mid \mathbf{4}$},
              edge label={node[midway,left,xshift=30pt,font=\scriptsize\sffamily]{U}}
              [{$\dots$}, 
                edge label={node[midway,left,font=\scriptsize\sffamily]{STEP}},
                            [{$u_{18}:\ \neg c\mid \mathbf{9}$}, myleaf={\cmark\ \textsf{EMPTY}}]
              ]
            ]
            [{$\dots$} 
            ]
          ]
        ]
      ]
    ]
  ]
\end{forest}
\caption{Tableau with the \oldjump{} rule for formula $\varphi'_S$.}
\label{fig:unsound-tree}
\end{figure}

The \oldjump{} rule also undermines the completeness of the tableau.
Consider the following formula:
\[
\varphi_C =
\suntil{[0, 10]}{a}{(b \land \glob{[20, 30]} c)} \land \glob{[0, 27]} \neg c \land \glob{[10, 10]} \neg b
\]
This formula is satisfiable, as shown by an accepting branch of the basic tableau in Fig.~\ref{fig:incomplete-tree}.
In satisfying signals, the $\suntilop$ operator must have its right-hand operand satisfied at times 8 or 9.
In fact, formula $\glob{[20, 30]} c$ conflicts with $\glob{[0, 27]} \neg c$ if it is extracted at a time in between 0 and 7, while $b$ conflicts with $\glob{[10, 10]} \neg b$ at time 10.

Thus, the tableau branch that tries to satisfy the $\suntilop$ operator at time 0 is rejected, while the one that postpones its satisfaction remains.
The \oldjump{} rule then jumps directly to time 10, where the $\suntilop$ operator must be satisfied, but it cannot: by skipping nodes with times 8 and 9, the \oldjump{} rule prevents the tableau from finding two possible accepting branches.
Thus, the \oldjump{} rule makes the tableau incomplete, as it may declare unsatisfiable formulas that are actually satisfiable.

By contrast, with our \textsf{JUMP} rule, in each node $u$ with $\suntil{[0, 10]}{a}{(b \land \glob{[20, 30]} c)} \in \Gamma(u)$ we have $(t(u)+20, t(u)+30) \in M(u)$,
and $(0,27) \in S(u)$, due to $\glob{[0, 27]} \neg c$.
Since for all $t \in [0,7]$ we have $(t+20, t+30) \cap (0,27) \neq \emptyset$,
the \textsf{JUMP} rule cannot be applied until a node with time 8 is reached.
In such a node the expansion rule for the $\suntilop$ operator is applied, generating an accepting branch.

\begin{figure}
\centering
\begin{forest}
  for tree={
    myleaf/.style={label=below:{\strut#1}},
    l sep=4mm, s sep=4mm
  },
  [{$u_{0}:\ \suntil{[0,10]}{a}{(b \land \glob{[20,30]} c)} \land \glob{[0,27]} \neg c \land \glob{[10,10]} \neg b\ \mid\ \mathbf{0}$}
    [{$\dots$},
      [{$u_{1}:\ \glob{[10,10]} \neg b,\ \glob{[8,27]} \neg c,\ \suntil{[8,10]}{a}{b \land \glob{[20,30]} c}\ \mid\ \mathbf{8}$},
        edge label={node[midway,left,font=\scriptsize\sffamily]{{STEP}}}
        [{$u_{2}:\ \glob{[10,10]} \neg b,\ \mglob{[8,27]} \neg c,\ \suntil{[8,10]}{a}{b \land \glob{[20,30]} c},\ \neg c\ \mid\ \mathbf{8}$}
          [{$u_{3}:\ \glob{[10,10]} \neg b,\ \mglob{[8,27]} \neg c,\ b \land \glob{[28,38]} c,\ \neg c\ \mid\ \mathbf{8}$}
            [{$\dots$},
              [{$u_{4}:\ \glob{[38,38]} c\ \mid\ \mathbf{38}$},
                edge label={node[midway,left,font=\scriptsize\sffamily]{{STEP}}}
                [{$u_{5}:\ c\ \mid\ \mathbf{38}$}, myleaf={\cmark\ \textsf{EMPTY}}]
              ]
            ]
          ]
          [{$\dots$}]
        ]
      ]
    ]
  ]
\end{forest}
\caption{Part of an accepting branch of the basic tableau for formula $\varphi_C$.}
\label{fig:incomplete-tree}
\end{figure}
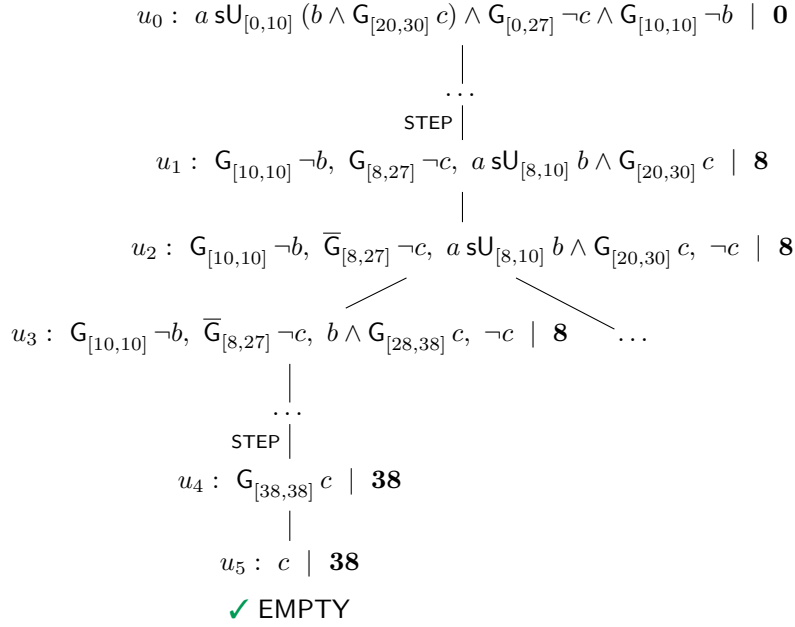

\section{Theorem Proofs}
\label{app:proofs}

\begin{proof}[Proof of Lemma~\ref{lemma:te-limit}]
    If $\rho$ does not appear within a temporal operator, then it may appear only in the labels of descendants $v$ of $u$ generated by the expansion rules for the $\land$ and $\lor$ operators.
    Such rules generate children nodes with the same time as the parent, hence $t(v) = t(u)$.
    Moreover, the \textsf{STEP} rule drops all non-temporal operators and propositions, hence none of the nodes it produces may contain $\rho$.

    We prove the case in which $\rho$ appears within a temporal operator only for the $t(v) \leq r$ side, the other one ($l \leq t(v)$) being analogous.
    The computation of $\tinv(\varphi)$ is recursive on the syntactic structure of $\varphi$, and $(l,r) \in \tinv(\varphi)$ results from a sequence of nested formulas terminating in $\rho$.
    We denote this sequence as $\pi = \langle \theta_n, \theta_{n-1}, \ldots, \theta_0 \rangle$, where:
    \begin{itemize}[nosep]
        \item $\theta_n = \varphi$;
        \item $\theta_{i}$ is the subformula of $\theta_{i+1}$ selected by the recursive application of $\tinv(\varphi)$ that yields $(l,r)$;
        \item $\theta_0 = \rho$.
    \end{itemize}
    For each $\theta_i$ in the path, let $(l_i, r_i) \in \tinv(\theta_i)$ denote the element contributed by $\theta_i$ to $(l, r)$ along the recursive computation, so that $(l_n, r_n) = (l, r)$ and $(l_0, r_0) = (0, 0)$.
    Furthermore, let $m$ be the index of the last temporal operator in the sequence:
    $\theta_m = \anybinop{I}{\theta_{m,1}}{\theta_{m,2}}$, with $\mathsf{B} \in \set{\suntilop, \sreleaseop}$,
    and for $m < i \leq n$, we have $\theta_i = \theta_{i,1} \circ \theta_{i,2}$ with $\circ \in \set{\land, \lor}$%
    \footnote{We assume $\varphi$ to be in strict normal form, so temporal operators cannot be nested in a negation.}.

    We now prove the claim that for each formula $\theta_i$ with $0 \leq i < m$, and any node $u_i$ such that $\exp^{t(u_i)}(\theta_i) \in \Gamma(u_i)$, and any node $u_i'$ in the subtree rooted at $u_i$, if we have $\rho \in \Gamma(u_i')$ that is an occurrence derived from the expansion of $\theta_i$ along $\pi$, then $t(u_i') \leq t(u_i) + r_i$.
    We proceed by induction on $i$, with cases on the shape of $\theta_i$.

    For the base case we have $\theta_0 = \rho$, and $\tinv(\theta_0) = \{(0,0)\}$, so $r_0 = 0$.
    Since no expansion rules can be applied to $\rho$, and the \textsf{STEP} rule does not propagate $\rho$ to its child, $\rho$ may only appear in descendants of $u_0$ that do not result from an application of \textsf{STEP}, thus only in $u_0$ itself or nodes $u'_0$ such that $t(u_0') = t(u_0)$.
    Hence, $t(u_0') = t(u_0) \leq t(u_0) + r_0$, and the claim holds.

    The inductive step is by cases on the shape of $\theta_i$:
    \begin{itemize}
        \item If $\theta_i = \neg \theta'_{i-1}$, then $\tinv(\theta_i) = \tinv(\theta'_{i-1})$, so $r_i = r_{i-1}$, and the claim follows by the induction hypothesis.
        \item Let $\theta_i = \theta_{i,1} \lor \theta_{i,2}$ and, w.l.o.g., $\theta_{i-1} = \theta_{i,1}$ (the other case is analogous).
              Since $\tinv(\theta_i) = \tinv(\theta_{i,1}) \cup \tinv(\theta_{i,2})$, we have $r_i = r_{i-1}$.
              According to Table~\ref{tab:expansion-rules}, the root $u_i$ of the subtree in which $\theta_i$ appears has two children $u_{i,1} = u_{i-1}$ and $u_{i,2}$ such that $\exp^{t(u_i)}(\theta_{i,1}) \in \Gamma(u_{i,1})$ and $\theta_i \not\in \Gamma(u_{i,1})$; $\exp^{t(u_i)}(\theta_{i,2}) \in \Gamma(u_{i,2})$ and $\theta_i \not\in \Gamma(u_{i,2})$; and $t(u_{i-1}) = t(u_{i,1}) = t(u_{i,2}) = t(u_i)$.
              Node $u'_i$ may only appear in the subtree rooted in $u_{i,1}$.
              Hence, by the induction hypothesis, $t(u'_i) \leq t(u_{i,1}) + r_{i-1} = t(u_i) + r_i$, and the claim holds.
        \item The case for $\theta_i = \theta_{i,1} \land \theta_{i,2}$ is analogous.
        \item Let $\theta_i = \suntil{[a,b]}{\theta_{i,1}}{\theta_{i,2}}$.
              According to the operator semantics and Table~\ref{tab:expansion-rules}, since $\exp^{t(u_i)}(\theta_i) \in \Gamma(u_i)$, any branch that contains $u_i$ must contain%
              \footnote{Unless the branch is rejected before $u_{i,2}$ appears. In this case, however, the claim is trivially true.}
              a node $u_{i,2}$ such that $t(u_{i,2}) \in [a + t(u_i), b + t(u_i)]$ and $\exp^{t(u_{i,2})}(\theta_{i,2}) \in \Gamma(u_{i,2})$,
              and for all time instants $t'$ with $a + t(u_i) \leq t' < t(u_{i,2})$ there must be a node $u_{t'}$ such that $t(u_{t'}) = t'$ and $\exp^{t(u_{t'})}(\theta_{i,1}) \in \Gamma(u_{t'})$.
              Assume w.l.o.g.\ that $\theta_{i-1} = \theta_{i,1}$ (the other case is analogous).
              Consider any such node $u_{t'}$.
              By the induction hypothesis, for all nodes $v$ in the subtree rooted in $u_{t'}$, if we have $\rho \in \Gamma(v)$, then $t(v) \leq t(u_{t'}) + r_{i-1} = t' + r_{i-1}$.
              By the definition of $\tinv(\theta_i)$, we have $r_i = r_{i-1} + b - 1$, and since $t' < b + t(u_i)$ we have $t(v) \leq t' + r_{i-1} < b + t(u_i) + r_{i-1} = t(u_i) + r_i$, and the claim holds.

              The $\sreleaseop$ operator can be treated analogously.
    \end{itemize}

    We now come to $\theta_i$ with $m \leq i \leq n$, for which we prove that for any node $u_i$ such that $\theta_i \in \Gamma(u_i)$,
    and any node $u_i'$ in the subtree rooted at $u_i$, if we have $\rho \in \Gamma(u_i')$ that is an occurrence derived from the expansion of $\theta_i$ along $\pi$,
    then $t(u_i') \leq r_i$.

    Let $\theta_m = \suntil{[a,b]}{\theta_{m,1}}{\theta_{m,2}}$ (the case for $\sreleaseop$ is analogous).
    The argument is similar to the inductive step of the previous claim, except $\theta_m \in \Gamma(u_m)$, instead of $\exp^{t(u_m)}(\theta_m) \in \Gamma(u_m)$.
    Thus, any branch that contains $u_m$ must contain a node $u_{m,2}$ such that $t(u_{m,2}) \in [a, b]$ and $\exp^{t(u_{m,2})}(\theta_{m,2}) \in \Gamma(u_{m,2})$,
    and for all time instants $t'$ with $a \leq t' < t(u_{m,2})$ there must be a node $u_{t'}$ such that $t(u_{t'}) = t'$ and $\exp^{t(u_{t'})}(\theta_{m,1}) \in \Gamma(u_{t'})$.
    We assume w.l.o.g.\ that $\theta_{m-1} = \theta_{m,1}$: our previous claim holds for it.
    Hence, consider any such node $u_{t'}$.
    For all nodes $v$ in the subtree rooted in $u_{t'}$, if we have $\rho \in \Gamma(v)$, then $t(v) \leq t(u_{t'}) + r_{m-1} = t' + r_{m-1}$.
    By the definition of $\tinv(\theta_m)$, we have $r_m = r_{m-1} + b - 1$, and since $t' < b$ we have $t(v) \leq t' + r_{m-1} \leq b - 1 + r_{m-1} = r_m$, and the claim holds.

    To conclude the proof, it suffices to notice that for $m < i \leq n$, the top-level operator in $\theta_i$ is $\land$ or $\lor$.
    Hence, $r_i = r_m$, and $\rho$ appears in descendants $v$ of $u_m$ such that $t(v) \leq r_m = r_i$.
\end{proof}

\begin{proof}[Proof of Lemma~\ref{lemma:jump-sound}]
    We proceed by contradiction. Assume that no branch below $u$ in the basic tableau is accepted, while the subtree rooted at $u'$ contains an accepted branch.

    If the \textsf{STEP} rule is applied, then $u'$ is also the successor of $u$ in the basic tableau, so the accepted branch below $u'$ is already a branch below $u$ in the basic tableau, contradicting the assumption. 
    Hence, in the following we assume that $u'$ is obtained by applying the \textsf{JUMP} rule. Let $k = k^{\ast}(u)$ be the jump size, so that $t(u') = t(u) + k$.

    Since $\Gamma(u)$ is consistent (otherwise $u$ would be rejected by the local inconsistency check), any rejection in the basic tableau must arise from further expansions along branches below $u$.

    We focus on the branch of the basic tableau where we postpone the satisfaction of temporal operators---i.e., obtained by always following the postponement child $\Gamma_\varphi(u_p)$ whenever an active temporal operator is expanded---and apply the \textsf{STEP} rule until we reach a poised node $v$ with time $t(u) + j$. 
    Since no branch below $u$ is accepted by assumption, this branch must be rejected.

    The only structural difference between the basic tableau below $u$ and the jump tableau below $u'$ is that the basic tableau explicitly generates intermediate nodes at time instants $t(u) + 1, \dots, t(u') - 1$, while the jump tableau skips them. 
    Since $k \leq k(u)$, no temporal operator becomes active or expires at a skipped instant.
    Hence, divergence in satisfiability between the two trees can only occur if extracting an invariant from some active temporal operator at an intermediate time step $t(u) + j$, for $1 \leq j < k$, introduces atomic constraints that become inconsistent with constraints derived from some formula in $\Gamma(u)$.
    We now show that no such new contradictions arise.

    Let $\rho$ denote an atomic proposition arising from the extraction of an invariant of some marked temporal operator at an intermediate time, and let $\sigma$ denote an atomic proposition from another formula $\psi$ in $\Gamma(u)$.

    Let $(l_n,r_n) \in N(u)$ be the validity interval (from $\tinv(\exp^{t(u)}(\varphi_1))$) of the invariant producing $\rho$, and let $(l_o,r_o) \in O(u)$ be the validity interval of $\tinv(\psi)$ terminating on the atomic proposition $\sigma$.
    Note that formulas with an active parent are excluded from $O(u)$ because the validity intervals of any derived proposition $\sigma$ are contained in those of their active parent.
    This redundancy is a consequence of Lemma~\ref{lemma:te-limit}, which indirectly shows that the validity intervals of propositions in nested subformulas are always included in the ones of their active ancestors.

    The soundness condition of the \textsf{JUMP} rule, namely \eqref{eq:sound-condition}, guarantees that for all such pairs, either $r_n < l_o$ or $r_o < l_n$.

    \begin{description}
        \item[Case $r_n < l_o$] 
        First, suppose that $\rho$ is nested into a temporal operator within $\varphi_1$.
        Since $j < k \leq k^{\ast}_{\text{sound}}(u) \le l_o - r_n$, we have $r_n + j < l_o$, where $r_n + j$ is the upper validity bound of $\exp^{t(u) + j}(\varphi_1)$.
        Applying Lemma~\ref{lemma:te-limit} to $\exp^{t(u) + j}(\varphi_1)$, any occurrence of $\rho$ can only appear in a node $w_{\rho}$ with $t(w_{\rho}) \le r_n + j$.
        Applying Lemma~\ref{lemma:te-limit} to $\psi$, any occurrence of $\sigma$ derived from the expansion of $\psi$ can only appear in a node $w_{\sigma}$ with $t(w_{\sigma}) \ge l_o$.
        Hence $t(w_{\rho}) \le r_n + j < l_o \le t(w_{\sigma})$, so $\rho$ and $\sigma$ cannot occur together in any node and no inconsistency can arise.

        If $\rho$ is not nested into any temporal operator in $\varphi_1$, then by Lemma~\ref{lemma:te-limit} we have $t(w_{\rho}) = t(u) + j \leq r_n + j < l_o \le t(w_{\sigma})$.
        Note that, by the definition of $O(u)$, $\sigma$ is always nested into at least one temporal operator.
        
        \item[Case $r_o < l_n$]
        Note that $\rho$ must be nested into a temporal operator within $\varphi_1$. Otherwise, we would have $l_n = t(u) > r_o$, meaning that $\sigma$ could only occur in the past.

        Applying Lemma~\ref{lemma:te-limit} to $\psi$, any occurrence of $\sigma$ derived from the expansion of $\psi$ can only appear in a node $w_{\sigma}$ with $t(w_{\sigma}) \leq r_o$. 
        Applying Lemma~\ref{lemma:te-limit} to $\exp^{t(u) + j}(\varphi_1)$, any occurrence of $\rho$ derived from its expansion can only appear in a node $w_{\rho}$ with $t(w_{\rho}) \geq l_n + j$.
        Hence $t(w_{\sigma}) \leq r_o < l_n < l_n + j \leq t(w_{\rho})$, so $\rho$ and $\sigma$ cannot co-occur in any node and no inconsistency can arise.
    \end{description}

    We have shown that no inconsistency can be triggered by invariant extractions at any skipped time $t(u) + j$. 
    Consequently, the postponement branch of the basic tableau cannot be rejected before time $t(u)+k$. At time $t(u)+k$, this branch preserves the same temporal obligations as $u'$.
    Since the subtree rooted at $u'$ contains an accepted branch, the basic tableau rooted at $u$ admits an accepted branch too, contradicting the initial assumptions.
\end{proof}

\begin{proof}[Proof of Theorem~\ref{thm:jump-sound}]
    Let $\mathcal{T}_J$ be the tableau constructed using the JUMP rule, and assume it contains an accepted branch $\mathbf{u} = u_0, \dots, u_n$. We consider the poised nodes in the accepted branch $\mathbf{p} \subseteq \mathbf{u}$. 

    Take a poised node $p_i$ where the JUMP rule is applied, resulting in a node $v$, so that $t(v) = t(p_i) + k^{\ast}(p_i)$. By Lemma~\ref{lemma:jump-sound}, if $k^{\ast}(p_i) < k^{\ast}_{\text{sound}}(p_i)$, the basic tableau subtree rooted at $p_i$ must also contain an accepted branch. In the case where $k^{\ast}(p_i) = k^{\ast}_{\text{sound}}(p_i)$, time instant $t(v)$ will be the first potential time instant where a conflict could arise, but this node is checked via the standard consistency checks both in the basic tableau and the one using the \textsf{JUMP} rule.
\end{proof}

\begin{proof}[Proof of Theorem~\ref{thm:jump-complete}]
    Consider a poised node $u$ such that the \textsf{JUMP} rule is applied to it, and denote the resulting node as $u'$.
    
    Let us prove the claim by contradiction.
    Assume that the basic tableau contains an accepting branch, but all branches originating from $u$ in the tableau with the \textsf{JUMP} rule are rejected.
    This means that every branch generated from the temporal obligations in $\Gamma(u')$ is eventually rejected. 

    If a step condition is triggered, then $u'$ is the successor of $u$ in the basic tableau as well, so the tableau with the \textsf{JUMP} rule would be accepted as well, contradicting the hypothesis. Hence, in the following, we assume that $u'$ is obtained by applying the \textsf{JUMP} with a jump size $k = k^{\ast}(u)$, so that $t(u') = t(u) + k$.

    The only possibility for the basic tableau to contain an accepted branch is that, at some intermediate node $w$ with $t(u) < t(w) < t(u')$, the accepting branch selected the satisfaction of a temporal formula $\psi \in \Gamma(w)$ (through expansion rule $\Gamma_{\psi}(w_1)$) that, if postponed, would otherwise be propagated to $\Gamma(u')$, with $\psi = \suntil{[a,b]}{\varphi_1}{\varphi_2}$ or $\psi = \srelease{[a,b]}{\varphi_2}{\varphi_1}$.
    Let $j = t(w) - t(u)$ and note that $1 \leq j < k$. 

    Since $k \leq k(u)$, $\psi$ must be already active in $u$.    
    Consider node label $\Gamma_A := (\Gamma(u) \setminus \set{\psi}) \cup \set{\exp^{t(u)}(\varphi_2)}$, obtained by satisfying $\psi$ already at $u$, and node label $\Gamma_B^j := (\Gamma(w) \setminus \set{\psi}) \cup \set{\exp^{t(u)+j}(\varphi_2)}$, obtained by satisfying $\psi$ at $w$.
    We show that if the subtree rooted at node labeled $\Gamma_A$ is rejected, then the subtree rooted at node labeled $\Gamma_B^j$ is rejected as well, contradicting that node $w$ is part of an accepting branch.
    
    Since the subtree rooted at $\Gamma_A$ is rejected (as it is explored also by the jump tableau), every branch terminates in a nodes that contains a contradiction, or $\neg \top$.
    Let $\mathcal{V}$ be the set of rejected leaf nodes. For each $v \in \mathcal{V}$, let $C^v \subseteq \Gamma(v)$ be a minimal inconsistent subset of atomic constraints, with $C^v \neq \emptyset$. 

    Partition $C^v$ into $C^v_m \cup C^v_s$, where $C^v_m$ contains propositions derived from the extracted target $\varphi_2$ and $C^v_s$ those derived from $\Gamma(u) \setminus \set{\psi}$. 

    If $C^v_m = \emptyset$ the conflict lies entirely within propositions in $\Gamma(u) \setminus \set{\psi}$, while if $C^v_s = \emptyset$ it lies entirely within the extracted target propositions. In both cases, conflict would be independent of the time at which the target is extracted, thus would arise in $\Gamma_B^j$ as well, and the claim holds.
    The case in which rejection is due to $\neg \top$ being in $C^v_m$ or $C^v_s$ is also analogous.

    We focus on the remaining case $C^v_m \neq \emptyset$ and $C^v_s \neq \emptyset$, and $\neg \top \not\in C^v_m \cup C^v_s$.
    Choose $\sigma \in C^v_s$ and $\rho \in C^v_m$. Note that the constraints $\sigma$ and $\rho$ belong to the same minimal inconsistent set $C^v$, but they are not necessarily pairwise inconsistent.
  
    We analyze the source of the obstacle proposition $\sigma$:
    \begin{itemize}
        \item[(i)]
        If $\sigma \in \Gamma(u)$ has no active parent, then $s = (t(u), t(u)) \in S(u)$.
        
        If $\rho \in \Gamma(u)$, then by Lemma~\ref{lemma:te-limit} there exists $(l_\rho, r_\rho) \in \tinv(\exp^{t(u)}(\varphi_2))$ such that $l_\rho = t(u)$ and $m = (t(u), r_\rho) \in M(u)$.
        If $\rho$ is not nested into any temporal operator, then $m = (t(u), t(u)) \in M(u)$.

        Hence the intersection between $m$ and $s$ contains $t(u)$, enforcing a \textsf{STEP} rule due to \eqref{eq:complete-condition}.

        \item[(ii)] If $\sigma \in \Gamma(u)$ has an active parent, we again show that if propositions $\rho$ and $\sigma$ occur in the same leaf $v$, by Lemma~\ref{lemma:te-limit} their validity intervals intersect.
        
        In fact, if $\rho$ is nested into a temporal operator in $\varphi_2$, we have $m = (l_\rho, r_\rho) \in M(u)$ with $l_\rho \leq t(v) \leq r_\rho$.
        If $\rho$ is not nested into any temporal operator in $\varphi_2$, then we must have $t(v) = t(u)$, and $m = (t(u)+0, t(u)+0) = (t(v), t(v)) \in M(u)$.

        Moreover, $\Gamma(u)$ must contain a formula $\theta$ such that $\theta$ has no active parent, and $\sigma$ appears in $\theta$ (considering its identifier).
        Since $\theta$ has no active parent, its validity intervals are included in $S(u)$, and by applying Lemma~\ref{lemma:te-limit} we have $s = (l_\sigma, r_\sigma) \in S(u)$ such that $l_\sigma \leq t(v) \leq r_\sigma$.
        Then $\set{t(v)} \in m \cap s \neq \emptyset$, and a \textsf{STEP} rule would be executed, contradicting the possibility of making a jump in the first place.
    \end{itemize}

    Since the node labeled with $\Gamma_A$ is also explored by the jump tableau, and by the premise it is rejected, the node labeled with $\Gamma_B^j$ will be rejected too contradicting the assumption that $w$ is part of an accepting branch of the basic tableau and proving the theorem.
\end{proof}

\begin{corollary}
The \textsf{JUMP} rule with the distinct-variable optimization (\S\ref{sec:rule-application}) is sound and complete.
\end{corollary}
\begin{proof}[Proof (sketch)]
The claim follows from the proofs of Theorems~\ref{thm:jump-sound} and \ref{thm:jump-complete} and the following observations:
\begin{itemize}
\item Local inconsistencies in node labels are checked only among \emph{literals}, i.e., atomic propositions in positive or negated form.
    If such propositions are Boolean variables, their conjunction may only be unsatisfiable if the same variable appears both in positive and negated form.
    Thus, two such literals that contain different variables or are the same literal may never cause a node to be rejected by the local consistency check.
\item If a node contains real-valued constraints, it might be the case that, although two constraints do not share any variables, their conjunction with a larger set of constraints is unsatisfiable.
    In this case we use the fact that, if a set of constraints is unsatisfiable, then at least two of them must share a variable.

    Suppose, by contradiction, that the tableau that uses the optimization has an accepting branch, and the one without it has none.
    Let $v$ be a node that appears only in the unoptimized tableau and is rejected due to a contradiction.
    Node $v$ must descend from a poised node $u$ in which the guard \eqref{eq:sound-condition} allows the \textsf{JUMP} rule only when using the optimization, or the $k^{\ast}_{\text{sound}}(u)$ limit is larger when using the optimization, thus allowing for a longer jump.
    However, since $v$ is rejected, $\Gamma(v)$ must contain two real-valued constraints $\rho$ and $\sigma$ that share the same variable.
    Let $\varphi, \psi \in \Gamma(u)$ be the two formulas from which $\rho$ and $\sigma$ respectively derive.
    By Lemma~\ref{lemma:te-limit}, both $\tinv(\varphi)$ and $\tinv(\psi)$ must contain a validity interval that includes $t(v)$.
    This can be easily shown to contradict the fact that the optimized tableau jumps when the unoptimized does not, or that it makes a longer jump.

    A similar argument can be made for completeness along the lines of the proof of Theorem~\ref{thm:jump-complete}.
\end{itemize}
\end{proof}

\section{Random Benchmark Generation}
\label{app:stl-bench}

The \emph{Random-STL} and \emph{Random0-STL} benchmark suites were generated using a stochastic recursive expansion algorithm. The generator is implemented in Rust using the set of functions provided by the \emph{formula} module of the \stlsat{} tool.

\subsection{Generation Algorithm}

The generator produces formulas whose satisfiability is not known beforehand and is ultimately determined by the solving engines.
A formula is created as the conjunction of $j$ independently generated sub-formulas. Each of these sub-formulas is produced by a recursive function.

\paragraph{Recursive Generation}

The generation of a branch is terminated once an atomic proposition is selected (as it has no children). 
The average formula depth is controlled by a parameter $p_{\text{stop\_base}}$, representing the probability of selecting an atomic proposition. 
This probability grows linearly with the current tree depth $d$ such that $P_\text{stop} = \min(1, p_{\text{stop\_base}} \times d)$.

\paragraph{Temporal Operator}

Whenever the recursion continues, the algorithm selects between a temporal or a non-temporal operator based on a parameter $p_{\text{temp}}$.
Whenever a temporal operator is selected, the operator is chosen uniformly from the set $\set{\evenop, \globop, \untilop, \releaseop}$. The corresponding interval $I = [a,b]$ is computed to ensure the global temporal horizon is not exceeded. Let $h$ denote the current accumulated horizon at the given tree depth, $l$ the maximum global horizon parameter, and $i$ the maximum allowable interval length parameter. The maximum permissible time interval length is $\Delta_{\max} = \min(i, l - h)$.

For the \emph{Random-STL} benchmark, the interval $I$ is generated by selecting one of three modes with equal probability:
\begin{itemize}
    \item $a = 0$, and $b \sim \mathcal{U}(0, \Delta_{\max})$.
    \item $a \sim \mathcal{U}(0, \Delta_{\max})$, and $b \sim \mathcal{U}(a, \Delta_{\max})$.
    \item $a \sim \mathcal{U}(0, \Delta_{\max})$, and $b = l - h$.
\end{itemize}
For the \emph{Random0-STL} benchmark, the generation enforces the first mode (with $a = 0$ for all temporal operators).

\paragraph{Atomic Proposition}

When an atomic proposition is selected, it is sampled from a pre-generated pool of $c$ constraints. These constraints operate over the real variables $S = \{x_1, \dots, x_r\}$ and can take three forms, which are selected uniformly:
\begin{enumerate}
    \item \textbf{Simple Bounds:} $x_i \bowtie k$.
    \item \textbf{Difference Logic:} $x_i - x_j \bowtie k$ for $i \neq j$.
    \item \textbf{Linear Constraints:} $\sum_{x \in S'} x \bowtie k$, where $S' \subseteq S$ is a randomly selected subset.
\end{enumerate}
The relation $\bowtie$ is sampled from $\{<, \leq, >, \geq, =, \neq\}$, while the constant $k$ is a randomly generated rational number.

\subsection{Parameter Space and Suite Configuration}

The \emph{Random-STL} suite was generated by sweeping across all combinations of the parameter space. Table~\ref{tab:random_stl_params} presents the set of possible parameter values for each defined parameter. For each combination of parameters, $5$~formulas have been generated, for a total of $2000$~formulas.
The \emph{Random0-STL} variation uses identical parameter space value combinations but restricts all generated temporal operator intervals to the form $[0, b]$, as detailed in the previous section.

\begin{table}[h]
    \centering
    \begin{tabular}{ll}
        \toprule
        \textbf{Parameter} & \textbf{Values Sampled} \\
        \midrule
            Real Variables ($r$) & \set{5} \\
            Real Constraints ($c$) & \set{10} \\
            Max Horizon ($l$) & \set{100, 500, 1000, 5000, 10000} \\
            Max Interval Length ($i$) & \set{50, 1000} \\
            Stop Probability ($p_{\text{stop\_base}}$) & \set{0.05, 0.1, 0.15, 0.2, 0.25} \\
            Temporal Probability ($p_{\text{temp}}$) & \set{0.33, 0.5, 0.75, 0.95} \\
            Conjunctions ($j$) & \set{5, 10} \\
        \bottomrule
    \end{tabular}
    \caption{Parameter space for STL benchmark suites.}
    \label{tab:random_stl_params}
\end{table}

}{}

\end{document}